\documentclass[11pt]{article}

\usepackage{fullpage}
\headheight \baselineskip
\usepackage{tonsmod}
\usepackage{graphicx}
\usepackage{caption}
\usepackage{subcaption}
\usepackage{hyperref}

\usepackage{verbatim}
\usepackage{amssymb}
\usepackage{amsmath}
\usepackage[noend]{algpseudocode}
\usepackage{algorithm}
\usepackage{amsthm}
\usepackage{enumerate}

\usepackage{authblk}
\usepackage{xspace}

\newcommand{\fullv}[1]{#1}
\newcommand{\shortv}[1]{}

\newcommand{\commentout}[1]{}

\newcommand{\agents}{\mathcal{N}}

\newcommand{\hist}{\mathcal{H}}

\newcommand{\act}{\mathcal{A}}
\newcommand{\I}{\mathcal{I}}

\newcommand{\vsigma}{\vec{\sigma}}
\newcommand{\va}{\vec{a}}
\newcommand{\eu}[3]{\mathbb{E}^{#2}(u_{#1} \mid #3)}	
\newcommand{\pr}{Pr}						
\newcommand{\pra}[3]{\pr^{#2}(#1 \mid #3)}			
\newcommand{\good}{\textit{Good}}			
\newcommand{\bad}{\textit{Bad}}				
\newcommand{\pend}[2]{\mbox{AP}_{#1}^{#2}}

\newcommand{\acc}[3]{\mbox{RP}_{#1#2}^{#3}}
\newcommand{\dgr}{\mbox{deg}}
\newcommand{\oape}[1]{$#1$-OAPE}
\newcommand{\amu}[3]{\mu(#1\mid #2,#3)}

\newcommand{\G}{{\cal G}}

\newcommand{\po}[3]{\Phi^{#1}_{#2}(#3)}
\newcommand{\gen}{\mbox{\tiny gen}}
\newcommand{\val}{\mbox{\tiny val}}

\algblockdefx[Event]{UponEvent}{EndEvent}
   [1]{\textbf{Upon} #1}
   [0]{\textbf{End}}
\algblockdefx[EventP]{UponEventP}{EndEventP}
   [2]{\textbf{Upon} #1 $|$ #2}
   [0]{\textbf{End}}
\algblockdefx[Predicate]{Predicate}{EndPredicate}
   [2]{\textbf{Predicate} #1(#2)}
   [0]{\textbf{EndPredicate}}
\algblockdefx[After]{After}{EndAfter}
   [1]{\textbf{After} #1}
   [0]{\textbf{End}}

\usepackage{thmtools, thm-restate}

\declaretheorem{proposition}
\declaretheorem{lemma}
\newtheorem*{lemma*}{Lemma}

\begin{document}

\title{Accountability in Dynamic Networks}

\author{Xavier Vila\c{c}a\thanks{E-mail:
    {xvilaca@gsd.inesc-id.pt}. Address: INESC-ID, R. Alves Redol, 9,
    Lisbon, Portugal. Tel.: +351 213 100 359.}\ }
\author{Lu\'{i}s Rodrigues}

\affil{INESC-ID, Instituto Superior T\'{e}cnico, Universidade de Lisboa}

\date{} 
\maketitle

\thispagestyle{empty}

A variety of distributed protocols require pairs of neighbouring nodes
of a network to repeatedly interact in pairwise exchanges of messages
of mutual interest. Well known examples include
file-sharing systems\,\cite{Cohen:03} and gossip dissemination
protocols\,\cite{Li:06,Li:08}, among many other examples.
These protocols can operate in very diverse
settings such as wireless ad-hoc or peer-to-peer overlay networks.
These settings pose three main challenges. First, networks are
inherently dynamic, whether due to uncontrolled mobility or
maintenance of the overlay.  Second, since communication may be
costly, nodes may act rationally by not sending messages, while still
receiving messages from their neighbours.  Third, nodes may have
incomplete information about the network topology.  In this work, we
aim at gaining theoretical insight into how to persuade agents to
exchange messages in dynamic networks.

We take a game theoretical approach to determine necessary and
sufficient conditions under which we can persuade rational agents to
exchange messages, by holding them accountable for deviations with
punishments. Unlike previous
work~\cite{Mailath:07,Fainmesser:12,Fainmesser:12b,Fabrikant:03,Moscibroda:06,Li:06,Li:08,Guerraoui:10},
we do not assume the network is the result of rational behaviour or
that the topology changes according to some known probability
distribution, and we do not limit our analysis to one-shot pairwise
interactions. We make three contributions: (1) we provide a new game
theoretical model of repeated interactions in dynamic networks, where
agents have incomplete information of the topology, (2) we define a
new solution concept for this model, and (3) we identify necessary and
sufficient conditions for enforcing accountability, i.e., for
persuading agents to exchange messages in the aforementioned model.

Our results are of technical interest but also of practical
relevance. We show that we cannot enforce accountability
if the dynamic network does not allow for \emph{timely
punishments}. In practice, this means for instance that we cannot 
enforce accountability in some networks formed in file-sharing applications
such as Bittorrent\,\cite{Cohen:03}. We also show that for
applications such as secret exchange, where the benefits of the
exchanges significantly surpass the communication costs, timely punishments are enough
to enforce accountability. However, we cannot in general
enforce accountability if agents do not possess enough 
information about the network topology.
Nevertheless, we can enforce accountability in a wide variety 
of networks that satisfy 1-connectivity\,\cite{Kuhn:10} with minimal knowledge about the network topology.
This result provides a sound game theoretical foundation for the empirical 
results presented in~\cite{Li:06,Li:08,Guerraoui:10},  showing that we
can in fact enforce accountability in gossip dissemination with connected overlay
networks where nodes possess sufficient information about the network.

\shortv{\vspace{5mm}\noindent\textbf{Regular submission. Student paper.}}

\newpage
\setcounter{page}{1}

\section{Introduction}
A variety of protocols require pairs of neighbouring nodes of a network
to repeatedly interact in pairwise exchanges of messages that are of mutual interest to both parties.
Well known examples include file-sharing systems\,\cite{Cohen:03} and gossip
dissemination protocols\,\cite{Li:06,Li:08,Guerraoui:10}.
These protocols can operate in very diverse settings such as wireless ad-hoc or peer-to-peer overlay networks.
These settings pose three main challenges. First, networks are inherently dynamic,
whether due to uncontrolled mobility or maintenance of the overlay.
Second, since communication may be costly, nodes may act rationally by not sending messages,
while still receiving messages from their neighbours.
Third, nodes may have incomplete information about the network topology.
In this work, we aim at gaining theoretical insight into how to persuade agents
to exchange messages in dynamic networks.

A growing body of literature has taken a game theoretical
approach\,\cite{Osborne:94} to address rational behaviour in a variety
of distributed problems in static networks
(e.g.,\,\cite{Halpern:04,Abraham:06,Abraham:13,Afek:14,Wong:13}).
Nodes are viewed as being under the control of rational
agents\,\footnote{Henceforth, we use the designation agent when
  referring to both the node and the rational agents.} that seek to
maximize individual utilities.  In these lines, we use Game Theory to
address the problem of rational behaviour in pairwise exchanges of
messages over links of a dynamic network.  We consider that agents
obtain a benefit of receiving messages from their neighbours, and
incur costs for sending and receiving messages.  We are interested in
protocols that satisfy two properties: (1) agents exchange messages
in every pairwise interaction with neighbours, in order to provide
some useful service such as file-sharing and (2) the protocols are
equilibria according to some solution concept (e.g., Nash
equilibrium), so that no agent gains (increases its utility) by
deviating from the protocol. Protocols that
satisfy these properties are said to \emph{enforce accountability} in pairwise exchanges.

\vspace{1mm}\textit{Folk Theorems.} Results known as Folk
Theorems have shown that it is possible to enforce accountability by
holding agents accountable for deviations with punishments that
decrease their utility\,\cite{Mailath:07}, provided that (1) agents
want to participate in exchanges, i.e., the benefits of receiving
messages outweigh the communication costs, (2) agents perceive
interactions to occur infinitely often, which is possible if the
end-horizon of interactions is unknown\,\cite{Osborne:94}, and (3) a
monitoring infrastructure provides agents with sufficient information
regarding the past behaviour of other agents, so that any deviation
from the protocol may be promptly detected and punished.  In a
distributed system, agents may only learn about past behaviour of
other agents in messages received from their neighbours.  Therefore, a
monitoring infrastructure is constrained by the network.
Unfortunately, most proofs of Folk Theorems assume an exogenous
monitoring infrastructure, ignoring the constraints imposed by dynamic
networks. Some proofs of Folk Theorems take into account network
constraints\,\cite{Fainmesser:12,Vilaca:15}, but they assume that
networks are static or that agents know a probability distribution
over topologies at each point in time.  Such knowledge may not be
readily available to agents in our setting.  As we shall see,
this poses multiple challenges not addressed by existing proofs of Folk Theorems.

\vspace{1mm}\textit{Contributions.} 
We bridge the gap between existing proofs of Folk Theorems and the goal
of enforcing accountability in dynamic networks.
We make three contributions: (1) we provide a new game theoretical
model of repeated interactions in dynamic networks, where agents have
incomplete information of the topology, (2) we define a new solution
concept for this model, and (3) we identify necessary and sufficient
conditions for enforcing accountability in the aforementioned model.

\vspace{1mm}\textit{Game Theoretical Model.}
We consider that an \emph{adversary} selects the network
topologies\,\cite{Kuhn:10}.  We focus on networks that are not under
the control of the agents, exemplified by wireless ad-hoc networks and
overlays such as~\cite{Li:06,Li:08}.  To capture such exogenous
restrictions on the network, we assume that the adversary is oblivious
to the messages sent by agents, selecting the topologies for all times
prior to the beginning of the exchanges.  At each point in time,
agents learn partial information about the current topology, including
the identities of their current neighbours.  They also possess
information about the set $\G$ of dynamic networks that the adversary
may generate. In practice, the set $\G$ represents basic information known
by agents regarding the network structure. For instance, in an overlay network
designed to disseminate data in a reliable fashion, agents expect the
adversary to only generate connected graphs,
whereas in a more dynamic setting such as a wireless ad-hoc network,
$\G$ may contain a wider variety of non-connected graphs.

\vspace{1mm}\textit{Solution Concept.}  We define a solution concept
for our model, named $\G$-Oblivious Adversary Perfect
Equilibrium (\oape{\G}). To understand its definition, it is useful
to recall the notion of Nash equilibrium, which states that a protocol
is an equilibrium if it is a best response, i.e., no agent can
gain by deviating given that other agents also do not deviate.  This
definition does not suit our purposes for two reasons.  First, it does
not consider deviations \emph{off the equilibrium path}, i.e.,
deviations that occur after histories of messages where some agents
have deviated\,\cite{Osborne:94}. Thus, a Nash equilibrium may rely on
empty threats of punishments such as punishments that also cause some utility loss
to the punishers. Second, agents must know the probability
distribution of the adversary generating each topology. The first
issue is addressed in the literature by the notion of sequential
equilibrium (SE)\,\cite{Kreps:82}, which requires protocols to be best
responses after every history of messages. To address the second
issue, the notion of \oape{\G} refines SE by requiring protocols to be
best responses after every history, for any fixed network in $\G$ generated by
the adversary. In other words, we require that agents cannot gain
from deviating given that the adversary may generate any network in $\G$.

\vspace{1mm}\textit{Necessary and Sufficient Conditions.}
The main goal of the paper is to determine the degree of dynamism that we can tolerate
in order to enforce accountability in pairwise exchanges, i.e.,
we aim to determine necessary and sufficient restrictions on $\G$
for \oape{\G} protocols that enforce accountability to exist.
Our results show that the ability to enforce accountability
depends not only on $\G$ but also on the structure and 
utility of pairwise exchanges.

We determined the weakest set of restrictions on $\G$
for which we can enforce accountability. Our first result
shows that $\G$ must satisfy a property named \emph{timely punishments}.
Intuitively, this property states that for ever pairwise exchange between agents
$i$ and $j$, agent $j$ must always be able to communicate a deviation of $i$
to some third agent that later interacts with $i$, and thus is able to punish $i$ for that deviation.
This restriction is not met by some networks such as file-sharing overlays (e.g., Bittorrent\,\cite{Cohen:03}),
where users with similar interests interact frequently with each other
but only rarely with users with different interests.
Our next result provides a \oape{\G} protocol that enforces accountability,
assuming that $\G$ is restricted by timely punishments and that exchanges
are \emph{valuable}, that is, agents have a high benefit/cost ratio of
receiving/sending messages and neglect download costs.
Such type of exchanges occurs, for instance, when agents share small
but highly valuable secrets such as private keys\,\cite{Li:06,Li:08}.

In many cases, the assumption that pairwise exchanges are valuable is too restrictive.
For instance, in file-sharing, we may expect agents to be interested in exchanging
large files, but the cost of uploading such files is certainly non-negligible,
and the benefit-to-cost ratio of receiving/sending files may be small.
It is therefore of practical interest to understand whether we can enforce accountability in more
general settings. Our final two results identify a necessary and a sufficient condition to achieve this goal.
Namely, we show that $\G$ must satisfy, in addition to timely punishments,
a property that we call \emph{eventual distinguishability}. Roughly speaking, this property
states that whenever two (or more) agents may punish
$i$ from defecting in the past (towards some other agent $j$),  they
have the information necessary to coordinate their actions,
so that the total number of additional punishments for a single deviation of $i$ is never too large.
\commentout{
Specifically, at each point in time, every agent $i$ expects to be punished by other agents a number of times
that depends on past deviations. For each new deviation, the expected number of punishments
must increase at least one, or else $i$ gains by deviating. However, we show that it cannot increase by more than
one. Otherwise, the expected number of punishments may grow up to a point where
any additional punishment lies so far in the future that $i$ always prefers to deviate.
}Whether $\G$ satisfies eventual distinguishability depends both on the properties of graphs in $\G$
and on the information available to agents about the topology.
Following previous work in dynamic networks\,\cite{Kuhn:10},
we consider a class of connected networks where agents have knowledge
of the degree of their neighbours, which are formed by overlay networks such as~\cite{Li:06,Li:08,Guerraoui:10}.
We show that if graphs in $\G$ are contained in this class of connected networks,
then there is a protocol that enforces accountability in general pairwise exchanges.

\vspace{1mm}\textit{Summary of Results.}  Our results are summarized in Fig.~\ref{fig:results}.
We show that, in general, timely punishments (grey area) and eventual distinguishability (dotted area)
are necessary conditions for enforcing accountability, whereas connectivity is a sufficient condition (protocol $P1$).
In valuable pairwise exchanges, timely punishments are sufficient for enforcing accountability (protocol $P2$).

\begin{figure}
\centering
\begin{subfigure}{.4\textwidth}
 \centerline{\includegraphics[scale=0.4]{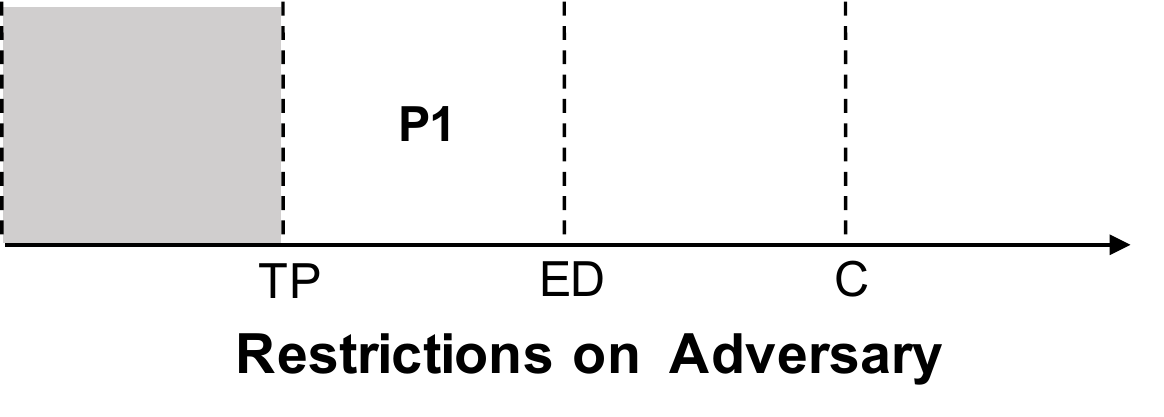}}
  \caption{Valuable Pairwise Exchanges (Sec.~\ref{sec:necessary})}
  \label{fig:val}
\end{subfigure}
\begin{subfigure}{.4\textwidth}
 \centerline{\includegraphics[scale=0.4]{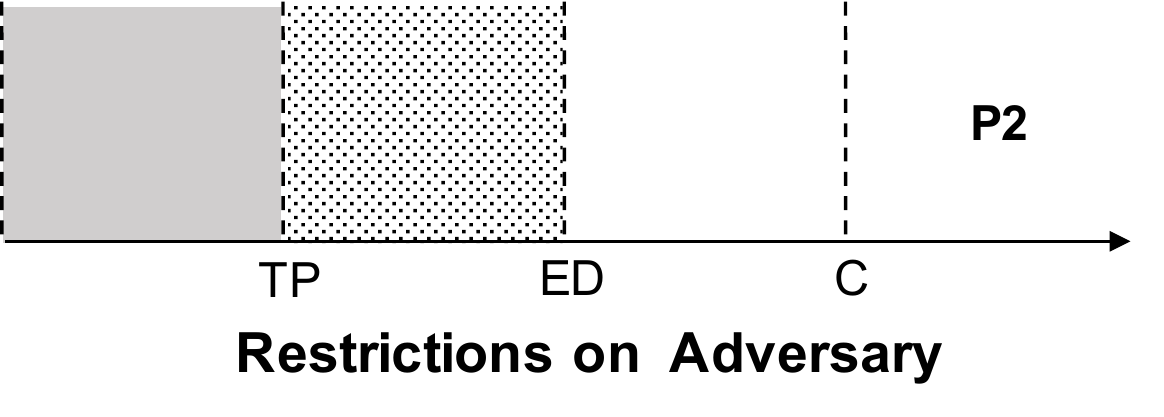}}
 \caption{General Pairwise Exchanges (Sec.~\ref{sec:sufficient})}
 \label{fig:gen}
\end{subfigure}
\caption{timely punishments (TP), eventual distinguishability (ED),  and connectivity (C).}
\label{fig:results}
\end{figure}

\vspace{1mm}\textit{Related work.}
Existing works have used Game Theory to analyse interactions in file sharing\,\cite{Jun:05,Neglia:07,Levin:08}
and gossip dissemination\,\cite{Li:06,Li:08}. Unlike these works,
we do not limit our analysis to one-shot pairwise interactions.
Our proposed model and results are more related to work in models of 
dynamic games and network formation games.
In the former, the structure of the game being repeated varies in each 
repetition according to a known probability distribution\,\cite{Mailath:07,Fainmesser:12,Fainmesser:12b}.
This captures repetitions of the same game where only the network topology 
varies. In the latter, the network is the outcome of the actions of rational agents\,\cite{Fabrikant:03,Moscibroda:06}.
To the best of our knowledge, none of these models captures network variations 
as being caused by an adversary oblivious to the actions of agents,
and they do not model incomplete information of the network topology.
Hence, they are not appropriate for modelling unpredictable changes
in the network such as physical topology changes. More importantly,
they do not address the challenges of devising distributed monitoring mechanisms.
Work in reputation systems such as~\cite{Feldman:04} have proposed and analysed
distributed systems that perform monitoring, but they do not prove Folk Theorems.

\section{Model}
\label{sec:model}
We consider a synchronous message passing system with 
reliable communication in a dynamic network.
The system entities are the rational agents, which send
messages over links of the network, and an oblivious adversary,
which selects the dynamic network at the beginning.
Specifically, time is divided into rounds. Prior to the execution of the protocol, 
the adversary selects an evolving graph $G$ that specifies
the communication graph $G^m$ of each round $m$.
Let $\G$ be the set of all evolving graphs, and let $\agents = \{1 \ldots n\}$ be the set of agents,
where $n$ is the number of agents.
As in~\cite{Abraham:06,Abraham:13}, we assume that edges are private:
agent $i$ can send a message to $j$ at round $m$ only if $j$ is a neighbour of $i$ in $G^m$,
and if $i$ sends that message, then exactly $j$ receives it; 
if $i$ has multiple neighbours, then $i$ can discriminate neighbours by sending different messages to each.
We also assume for simplicity that graphs are undirected.
Finally, we assume that each agent knows $n$ and
the identities of its neighbours in each round $m$, prior to sending messages in $m$.

\subsection{Pairwise Exchanges}
We consider an infinitely repeated game of pairwise exchanges of
messages between neighbouring agents.  Given $G \in \G$, at each round
$m$, every two neighbouring agents $i$ and $j$ have values of interest
to each other. The goal is to persuade $i$ and $j$ to share these
values plus some additional information required to monitor other agents.

\vspace{1mm}\textit{Actions and histories.} Agents $i$ and $j$
may exchange messages in one or more communication steps.
For simplicity, we abstract communication by considering a finite set of individual actions,
which capture most exchanges of interest,
namely (1) \emph{defection}, where $i$ omits
messages (and thus passively punishes $j$), (2) \emph{cooperation},
where $i$ sends its value plus monitoring information to $j$, and (3)
active \emph{punishment}, where $i$ sends messages while causing a utility loss to $j$
(for instance, if $i$ sends garbage instead of the
value\,\cite{Li:06}). Our results can easily be generalized to
arbitrary (finite) sets of individual actions.
In some of our results, we also consider that a punishment
can be proportional to some constant $c$,
i.e., can cause a utility loss proportional to $c$;
we also consider a fourth action of \emph{punishment avoidance},
where $i$ avoids the cost of a punishment but does not receive the value
from $j$. Both punishment avoidance and proportional punishments
are possible actions when multiple communication steps occur between agents;
later, we discuss possible implementations.

At every round $m$, $i$ and $j$ simultaneously follow an individual action, and are
only informed of each other's individual actions at the end of the
round. A round-$m$ \emph{action} $a_i$ of $i$ specifies the
individual actions of $i$ towards every neighbour.  A round-$m$
\emph{action profile} $\va$ specifies the round-$m$ actions followed
by all agents. When the adversary generates an evolving graph $G \in \G^*$,
repeated pairwise exchanges are characterized by a set $\hist(G)$ of \emph{histories}.
A round-$m$ history $h \in \hist(G)$ is a pair $((\va^{m'})_{m' < m},G)$ representing
the sequence of action profiles $\va^{m'}$ followed in rounds $m' < m$
and the evolving graph $G$ selected by the adversary.
In other words, $h$ represents
global information available immediately before agents follow
round-$m$ actions. A \emph{run} $r$ is a function mapping each round
$m$ to a round-$m$ history $r(m)$\,\footnote{For every $m' <m$, 
the sequence of actions in history $r(m')$ is a prefix of the sequence of actions in $r(m)$.}.

\fullv{Given a round-$m$ history $h$ and $G \in \G$, a protocol $\vsigma$ defines a probability
distribution $\pra{r}{\vsigma}{G,h}$ over every run $r$ compatible
with $h$ and $G$ (i.e., $r(m) = h$ and $G$ is the evolving graph in
$h$). Let $\pra{h}{\vsigma}{G}$ be the probability of history $h$
being realized conditioned on $G$ and agents following $\vsigma$ from
the beginning.  We say that $r$ is a run of $\vsigma$ after $h$ in $G$ if
$\pra{r(m')}{\vsigma}{G,h} > 0$ for every round $m' > m$.
As an abuse of notation, we consider that $\pra{r}{\vsigma}{G,h} > 0$.
We say that $r$ is a run of $\vsigma$ in $G$ if $h$ is the initial history.}

\vspace{1mm}\textit{Information and Strategies.}
Information available to agents can be divided into knowledge
about the game structure and private observations.
Regarding the former, we assume that the main information about the game
structure is common knowledge\,\footnote{Every agent knows this, knows that every agent knows this, and so on.}.
We also assume that there is a subset $\G^* \subseteq \G$
of evolving graphs that the adversary may generate and that is also common knowledge.
This does not preclude agents from believing that the adversary may generate any evolving graph (i.e., $\G^* = \G$);
however, we show that we need minimum restrictions on $\G^*$ to persuade agents to exchange messages.
Regarding the latter, given $G \in \G$ and round $m$,
every agent $i$ acquires some information about  $G^m$ prior to sending messages in $m$,
which includes the identity of the round-$m$ neighbours of $i$.
We represent this information as a set $\G_i^m$ of graphs, such that $G^m \in \G_i^m$
and, in every graph $\bar{G}$ in $\G_i^m$ agent $i$ obtains the same information about $\bar{G}^{m}$ and $G^m$.
For instance, if agents only know the identity of their neighbours, then $\G_i^m$ is the 
set of graphs where $i$ has the same set of neighbours as in $G^m$.

Given a round-$m$ history $h$, there is a private history $h_i$ of observations made by $i$ 
up to the beginning of round $m$, which include the sets $\G_i^{m'}$ for every round $m' \le m$,
the actions of $i$ in every round $m' < m$, and the individual actions
followed by neighbours of $i$ towards $i$ in every round $m' < m$.
We can associate to each private history $h_i$ a round-$m$ information set $I_i$
containing the histories that provide the same information as in $h_i$.
Given $G \in \G$, we denote by $\I_i(G)$ the set of information sets of $i$ compatible with $G$.
We make the standard assumption that agents have perfect recall, that is, 
agents recall all their observations made in the past.
A strategy $\sigma_i$ of agent $i$ corresponds to a distributed protocol.
Specifically, for each $G \in \G$ and information set $I_i \in \I_i(G)$,
$\sigma_i(\cdot \mid I_i)$ is a probability distribution over round-$m$ actions available to $i$.
We use the designation \emph{protocol} for strategy profiles $\vsigma$, which specify the strategy followed by every agent.

\vspace{1mm}\textit{Utility.} When agents follow an action profile $\va$, agent $i$ obtains
a utility $u_i(\va)$. This is the sum of the utilities obtained in interactions with each neighbour $j$
given as the difference between the benefits of receiving a value from $j$ and the costs of punishments and communication.
Specifically, whenever $j$ follows a cooperation or punishment action and $i$ does not avoid a punishment,
$i$ obtains a benefit $\beta$ of receiving $j$'s value. Regarding communication costs, 
$i$ incurs a normalized cost $1$ of sending messages in cooperation and punishment actions,
incurs no cost by defecting and avoiding punishments, and incurs the cost $\alpha$ of receiving messages from $j$.
Finally, a punishment of $j$ proportional to $c$ causes a utility loss of $c\pi$ to $i$,
where $\pi$ is the unitary cost per punishment.

We now define the (total) expected utility $u_i(\vsigma \mid G,I_i)$ of $i$ 
when agents follow a given protocol $\vsigma$,
conditional on some $G \in \G^*$ and on a round-$m$ information set $I_i$.
Given a round-$m$ history $h \in I_i$, $\vsigma$ defines a probability distribution over runs
compatible with $G$ and $h$. Therefore, we can compute the expected utility of $i$
in every future round $m' \ge m$ conditional on $G$ and $h$ as the expected value of $u_i(\va)$,
where the expectation is taken relative to the probability defined by $\vsigma$ 
of agents following action profile $\va$ in round $m'$.
Informally, $u_i(\vsigma \mid G,I_i)$ is the expected value of $u_i(\va)$
for every future round $m' \ge m$, as computed in round $m$.
Formally, we need to take into account both
(1) the probability of each $h \in I_i$ being realized
and (2) the effect of time on the value of future utilities.
Regarding (1), we consider a belief system $\mu$, specifying 
the probability $\amu{h}{G}{I_i}$ that $h$ is realized.
Regarding (2), we follow the standard approach\,\cite{Mailath:07}
of assuming that the value of future utilities decays over time with a discount factor $\delta \in (0,1)$,
so the expected utility in round $m'$ discounted to round $m$ is $\delta^{m'-m} u_i(\va)$.
Let $u_i(\vsigma \mid G,h)$ be the sum of the expected utilities $u_i(\va)$ discounted to 
$m$ for all future rounds $m' \ge m$, conditional on $G$ and $h$;
then, $u_i(\vsigma \mid G,I_i)$ is the expected value of $u_i(\vsigma \mid G,h)$,
where the expectation is taken relative to $\amu{h}{G}{I_i}$.

\subsection{Enforcing Accountability}
We say that a protocol $\vsigma$ \emph{enforces accountability} iff (1) in $\vsigma$,
agents always cooperate by exchanging their values, until some agent deviates,
and (2) $\vsigma$ is an equilibrium, so that no agent gains by deviating.
Regarding (2), we need to define a \emph{solution concept}, specifying exact conditions
under which $\vsigma$ is an equilibrium.

\vspace{1mm}\textit{Solution Concept.}
We define a new solution concept for our model,
which is a refinement of sequential equilibrium (SE)\,\cite{Kreps:82}.
In its original definition, a protocol $\vsigma^*$ is said to be a SE 
if there is a belief system $\mu^*$ consistent with $\vsigma^*$ (see below)
such that for every agent $i$ and information set $I_i$, $i$ cannot gain
by deviating from $\vsigma^*$ conditioning on $I_i$. We refine this definition by
also conditioning on every evolving graph $G \in \G^*$ that may be selected by the adversary.
Given this, if $\vsigma^*$ is an equilibrium, then no agent $i$ gains by deviating 
from $\sigma_i^*$, regardless of the evolving graph selected 
by the adversary and the observations made by $i$.
This is sufficient for ensuring that agents do not gain from not punishing
other agents or not forwarding monitoring information triggering additional punishments.

Formally, we say that $\vsigma^*$ is a $\G^*$-Oblivious Adversary Perfect Equilibrium (\oape{\G^*})
iff there is a belief system $\mu^*$ consistent with $\vsigma^*$ and $\G^*$ such that,
for every $G \in \G^*$, agent $i$, strategy $\sigma_i$, and information set $I_i \in \I_i(G)$,
$u_i(\vsigma^*  \mid G, I_i) \ge u_i((\sigma_i,\vsigma^*_{-i}) \mid G,I_i)$,
where the expectation is taken relative to $\mu^*$, and
$(\sigma_i,\vsigma^*_{-i})$ is the protocol where only $i$ deviates from $\sigma_i^*$ by following $\sigma_i$.
\shortv{The definition of belief system consistent with $\vsigma^*$ and $\G^*$
is very technical, and we leave it to Appendix~\ref{app:consistent}.}\fullv{The definition of consistent beliefs is identical in spirit to that proposed in the definition of sequential equilibrium\,\cite{Kreps:82}.
We say that a belief system $\mu^*$ is 
consistent with $\vsigma^*$ and $\G^*$ iff there is a sequence $\vsigma^1,\vsigma^2,\ldots$ 
of completely mixed protocols (attribute positive probability to every possible action at every information set) converging to $\vsigma^*$
such that for every $G \in \G^*$, agent $i$, information set $I_i \in \I_i$, and $h \in \I_i$,
$$\mu(h \mid G,I_i) = \lim_{c \to \infty} \pra{h}{\vsigma^c}{G}/\sum_{h' \in I_i} \pra{h'}{\vsigma^c}{G}.$$
}
Intuitively, given $G$, if $I_i$ is consistent with agents following $\vsigma^*$, then $i$
believes that agents followed $\vsigma^*$, otherwise, $i$ believes that they followed some alternative protocol.

\fullv{
\section{Properties of \oape{\G^*}}
We prove the One-Shot-Deviation principle and a property of consistent beliefs regarding a single
deviation from a \oape{\G^*}.

\subsection{One-Shot-Deviation Principle}
We find that the One-Shot-Deviation Principle holds for \oape{\G^*}\,\cite{Hendon:96}.
This principle states that, for a $\vsigma^*$ to be a \oape{\G^*},
it suffices that no agent $i$ can gain by performing a \emph{one-shot} deviation
from $\sigma_i^*$ at any information information set $I_i$,
where a one-shot deviation is a strategy where $i$ follows some action
not specified by $\sigma_i^*$ and follows $\sigma_i^*$ afterwards.
This principle will be useful later, so we provide a proof.
The proof is almost identical to~\cite{Hendon:96}, so we only include a sketch.
Let $\vsigma^*|_{I_i,a_i}$ be the protocol
identical to $\vsigma^*$, except $i$ deterministically follows $a_i$ at $I_i$.

\begin{proposition}
\textbf{One-Shot-Deviation Principle.} A protocol $\vsigma^* \in \Sigma$ is a \oape{\G^*}
if and only if there exists $\mu^*$ consistent with $\vsigma^*$ and $\G^*$ such that, for every $G \in \G^*$,
round $m$, agent $i$, round-m $I_i \in \I_i(G)$, and actions $a_i^*,a_i \in \act_i(G^m)$,
we have $u_i(\vsigma^*|_{I_i,a_i^*} \mid G,I_i) \ge u_i(\vsigma^*|_{I_i,a_i'} \mid G,I_i) \mid G,I_i)$.
\end{proposition}
\begin{proof}
(Sketch) We proceed as in~\cite{Hendon:96}, except that we fix $G \in \G^*$.
The implication is clear: since $\vsigma^*$ is a \oape{\G^*}, agent $i$ cannot
increase its expected utility by following $\vsigma^*|_{I_i,a_i'}$ instead of $\vsigma^*$,
and  $u_i(\vsigma^*|_{I_i,a_i^*} \mid G,I_i) = u_i(\vsigma^* \mid G, I_i)$.
As for the reverse implication, fix round-$m$ $I_i$. 
Suppose that the right-hand side of the proposition holds
and that there is $\sigma_i$ such that for some $\epsilon > 0$
\begin{equation}
\label{eq:one-dev}
u_i((\sigma_i,\vsigma^*_{-i}) \mid G,I_i) - u_i(\vsigma^* \mid G,I_i) = 2\epsilon.
\end{equation}

Define $m' > m$ such that $\delta^{m' - m} y n /(1-\delta) < \epsilon$.
Let $\sigma_i'$ be identical to $\sigma_i$ at every round-$m''$ information set for $m'' \le m'$,
but is identical to $\sigma_i^*$ at every round-$m''$ information set for $m'' > m'$.
It can be shown using the right-hand side and backwards induction that
$$u_i(\vsigma^* \mid G,I_i) \ge u_i((\sigma_i',\vsigma^*_{-i}) \mid G,I_i) \ge u_i((\sigma_i,\vsigma^*_{-i}) \mid G,I_i) -\epsilon.$$
This contradicts~(\ref{eq:one-dev}), proving the reverse implication and concluding the proof.
\end{proof}

\subsection{Single Deviation}
We now show that the definition of consistent belief systems implies that,
if an agent $i$ has been deviating from $\vsigma^*$
up to round $m$ and the round-$m$ information set $I_i$ is consistent
with only $i$ deviating from $\vsigma^*$, then $i$ believes that only he deviated.
Fix $G \in \G^*$, belief system $\mu^*$ consistent with $\vsigma^*$ and $\G^*$,
agent $i$, and round-$m$ information set $I_i \in \I_i(G)$.
We say that $I_i$ is consistent with only $i$ deviating from $\vsigma^*$ iff
there is a strategy $\sigma_i'$ and history $h \in I_j$
with $\pra{h}{(\sigma_i',\vsigma_{-i}^*)}{G} > 0$.
Let $\pra{\va}{\vsigma}{G,h}$ be the probability of agents following $\va$
conditioning on $h$ and $G$.

\begin{proposition}
\label{prop:onlydev}
If $I_i$ is consistent with only $i$ deviating from $\vsigma^*$,
then for all strategies $\sigma_i'$ consistent with $I_i$ and $G$ ($\pra{h'}{(\sigma_i',\vsigma_{-i}^*)}{G} > 0$ for some $h' \in I_i$)
and all histories $h \in I_i$ we have
$$\mu^*(h \mid G,I_i) = \frac{\pra{h}{(\sigma_i',\vsigma_{-i}^*)}{G}}{\sum_{h' \in I_i} \pra{h'}{(\sigma_i',\vsigma_{-i}^*)}{G}}.$$
\end{proposition}
\begin{proof}
Fix round-$m$ $h \in I_i$. Let $\vsigma^1,\vsigma^2,\ldots$ be a 
sequence of completely mixed protocols that converge to
$\vsigma^*$ such that
\begin{equation}
\label{eq:onlydev1}
\mu^*( h \mid G,I_i) = \lim_{c \to \infty} \frac{\pra{h}{\vsigma^c}{G}}{\sum_{h' \in I_i} \pra{h'}{\vsigma^c}{G}}.
\end{equation}
Since agents follow each action with an independent probability,
$\pra{h}{\vsigma^c}{G}$ is the product of a set of factors,
with one factor $\pra{h_j}{\sigma^c_j}{G}$ per agent $j$ corresponding to
the probability defined by $\vsigma^c$ of $j$ following the sequence of 
actions specified in $h$. Specifically, this is the product of $\sigma_j^c(a_j^{m'} \mid I_j^{m'})$
for every $m' < m$, where $a_j^{m'}$ is the round-$m'$ action of $j$ in $h$,
and $I_j^{m'}$ is the round-$m'$ information set corresponding to $h$.
The factor $\pra{h_i}{\sigma_i^c}{G}$ is the same for every $h \in I_i$.
Let $\pra{h_{-i}}{\vsigma^c_{-i}}{G}$ be the product of the
factors of agents different from $i$,
and let $I_i^*$ be the subset of histories from $I_i$ where only $i$ deviates from $\vsigma^*$.
Since $\vsigma^c$ converges to $\vsigma^*$, if $h' \in I_i^*$,
then $\lim_{c \to \infty}\pra{h'_{-i}}{\vsigma^c_{-i}}{G} = \pra{h'_{-i}}{\vsigma^*}{G}$,
and if $h' \in I_i \setminus I_i^*$, $\lim_{c \to \infty}\pra{h'_{-i}}{\vsigma^c_{-i}}{G} = 0$.
Therefore, for every $h \in \I_i$, we have 
$$
\begin{array}{lll}
&\mu^*( h \mid G,I_i) = 
\frac{\lim_{c \to \infty} \pra{h_{i}}{\sigma_i^c}{G}\pra{h_{-i}}{\vsigma^c_{-i}}{G}}{\lim_{c \to \infty}\sum_{h' \in I_i}\pra{h_{i}'}{\sigma_i^c}{G} \pra{h_{-i}'}{\vsigma^c_{-i}}{G}} &=\\
=&\lim_{c \to \infty} \frac{\pra{h_i}{\sigma_i^c}{G}}{{\pra{h_i}{\sigma_i^c}{G}}} \frac{\pra{h_{-i}}{\vsigma_{-i}^*}{G}}{\sum_{h' \in I_i^*} \pra{h_{-i}'}{\vsigma_{-i}^*}{G}} =
\frac{\pra{h_{-i}}{\vsigma_{-i}^*}{G}}{\sum_{h' \in I_i^*} \pra{h_{-i}'}{\vsigma_{-i}^*}{G}}& =\\
=&\frac{\pra{h_i}{\sigma_i'}{G}}{\pra{h_i}{\sigma_i'}{G}} \frac{\pra{h_{-i}}{\vsigma^*_{-i}}{G}}{\sum_{h' \in I_i^*} \pra{h_{-i}'}{\vsigma^*_{-i}}{G}} =
\frac{\pra{h}{(\sigma_i',\vsigma^*_{-i})}{G}}{\sum_{h' \in I_i} \pra{h'}{(\sigma_i',\vsigma^*_{-i})}{G}}.
\end{array}
$$

This proves the result.
\end{proof}
}

\section{Key Concepts}
In the proofs of our results, we identify multiple key concepts related
to the properties of the protocols and evolving graphs.
We summarize them here for future reference.

\vspace{1mm}\textit{Safe-bounded protocols.} It turns out that
it is relevant to our results to distinguish between safe and non-safe protocols,
and between bounded and non-bounded protocols.
Regarding the former, a protocol $\vsigma$ is \emph{safe} if in every interaction agents either cooperate or
punish each other, thus never omitting messages. Otherwise,
$\vsigma$ is \emph{non-safe}.
\shortv{In the full paper\,\cite{Vilacatr:16a},
we show that}\fullv{We show that}
in general non-safe protocols are not \oape{\G^*}.
For this reason, we restrict some of our positive and negative results to
safe protocols. Regarding the latter, a protocol $\vsigma$
is \emph{bounded} if the duration of punishments is bounded.
Specifically, every protocol can be represented as a state machine\,\cite{Osborne:94};
in a bounded protocol, (1) the number of states is finite, (2) some of those states are cooperation states where
all agents cooperate, and (3) starting from an arbitrary state, the time of convergence
to a cooperation state is bounded. Bounded protocols are simpler to analyse
and more useful in practice, since they ensure that memory is bounded, which is an important
requirement in dynamic networks. For these reasons,
we also restrict part of the analysis to bounded protocols.
\shortv{The full paper also generalizes our results for non-bounded protocols.}

\fullv{
Given an agent $i$, $G \in \G^*$, round-$m$ information set $I_i \in \I_i(G)$,
and $i$-edge $(j,m')$ with $m' > m$, $\vsigma$ defines a 
probability $\pra{j,m'}{\vsigma}{G,I_i}$ of $j$ punishing $i$ at round $m'$.
Let $\mathbb{E}^{\vsigma,\rho}[P_i \mid G]$ be the expected number of 
punishments of $i$ in the next $\rho-1$ rounds, defined as
$$\mathbb{E}^{\vsigma,\rho}[P_i \mid G,I_i] = \sum_{\mbox{$i$-edge }(j,m'): m < m' < m+\rho} \pra{j,m'}{\vsigma}{G,I_i}.$$

Proposition~\ref{prop:preserve} shows that if beliefs are consistent with $\vsigma$ and $\G^*$,
the expected number of punishments may be conserved over time: if
agent $i$ expects to be punished an additional number $c$ of times after deviating in round $m$, and agents follow $\vsigma^*$ afterwards,
then $i$ may still expect to be punished $c-k$ times after round $m+1$, where $k$ is the number of round-$m$ neighbours.

Fix a safe protocol $\vsigma$, $\rho > 0$, belief system $\mu^*$ consistent with $\vsigma$ and $\G^*$,
agent $i$, $G \in \G^*$, strategy $\sigma_i'$ where $i$ does not deviate from $\sigma_i$ in rounds $m' > m$, 
and round-$m$ information set $I_i \in \I_i(G)$ consistent with $\vsigma' = (\sigma_i',\vsigma_{-i})$ and $G$,
such that $\mathbb{E}^{\vsigma',\rho}[P_i \mid G,I_i] \ge \mathbb{E}^{\vsigma,\rho}[P_i \mid G,I_i] + c$ for some $c>0$.
Given round-$m$ history $h$ and action profile $\va$, let $(h,\va)$ be the round-$m+1$ history obtained by
appending $\va$ to the sequence of action profiles in $h$.

\begin{proposition}
\label{prop:preserve}
There is round-$m+1$ $I_i' \in \I_i(G)$ such that
$$\mathbb{E}^{\vsigma,\rho}[P_i \mid G,I_i'] \ge \mathbb{E}^{\vsigma,\rho}[P_i \mid G,I_i] + c - k,$$
where $k$ is the number of round-$m$ neighbours of $i$.
\end{proposition}
\begin{proof}
By Proposition~\ref{prop:onlydev}, we can write
\begin{equation}
\label{eq:preserve}
\begin{array}{lll}
  &\mathbb{E}^{\vsigma',\rho}[P_i \mid G,I_i] &= \\
=&\sum_{h \in I_i} \mu^*(h \mid G,I_i)\sum_{\va} \pra{\va}{\vsigma'}{G,h} (P_i(\va) + \sum_{r} \pra{r}{\vsigma}{G,(h,\va)} P_i(r,m)) &\le\\
\le & k + \sum_{h \in I_i} \mu^*(h \mid G,I_i)\sum_{\va} \pra{\va}{\vsigma'}{G,h} \sum_{r} \pra{r}{\vsigma}{G,(h,\va)} P_i(r,m) &=\\
=&k + \sum_{I_i' \in \I_i(G,I_i,m+1)} \sum_{(h,\va) \in I_i'} \pra{I_i'}{\vsigma'}{G,I_i} \mu^*((h,\va) \mid I_i') \sum_{r} \pra{r}{\vsigma}{G,(h,\va)} P_i(r,m) & =\\
=& k + \sum_{I_i' \in \I_i(G,I_i,m+1)}\pra{I_i'}{\vsigma'}{G,I_i} \mathbb{E}^{\vsigma,\rho}[P_i \mid G,I_i'].
\end{array}
\end{equation}
where $\I_i(G,I_i,m+1)$ is the set of round-$m+1$ information sets compatible with $I_i$ 
(for every $I_i' \in \I_i(G,I_i,m+1)$ and $(h,\va) \in I_i'$, $h \in I_i$),
$P_i(\va) \le k$ is the number of punishments of $i$ in $\va$,
$P_i(r,m)$ is the equivalent number of punishments that occur after round $m$ in run $r$,
and $\pra{I_i'}{\vsigma'}{G,I_i}$ is the probability of $I_i'$ being observed given $I_i$.
In the second step, we used the fact that summing over $h \in I_i$ and round-$m'$ action profile $\va$
is equivalent to summing over $I_i' \in \I_i(G,I_i,m+1)$ and $(h,\va) \in I_i'$, and multiplying
by the probability $\pra{I_i'}{\vsigma'}{G,I_i}\mu^*((h,\va) \mid G,I_i')$ of $I_i'$ being reached from $I_i$
and the history being $(h,\va)$. By Proposition~\ref{prop:onlydev}, we have
$$
\begin{array}{lll}
  &\pra{I_i'}{\vsigma'}{G,I_i} \mu^*((h,\va) \mid G,I_i') & =\\
=& \frac{\sum_{(h',\va') \in I_i'} \pra{(h',\va')}{\vsigma'}{G}}{\sum_{h' \in I_i} \pra{h'}{\vsigma'}{G}}
      \frac{\pra{(h,\va)}{\vsigma'}{G}}{\sum_{(h',\va') \in I_i'} \pra{(h',\va')}{\vsigma'}{G}} & =\\
=& \frac{\pra{h}{\vsigma'}{G}}{\sum_{h'  \in I_i} \pra{h'}{\vsigma'}{G}} \pra{\va}{\vsigma'}{G,h} & =\\
=& \mu^*(h \mid G,I_i) \pra{\va}{\vsigma'}{G,h}.
\end{array}
$$

Equation (\ref{eq:preserve}) implies that there is $I_i' \in \I_i(G,I_i,m+1)$ such that
$$
\mathbb{E}^{\vsigma,\rho}[P_i \mid G,I_i'] \ge  \mathbb{E}^{\vsigma',\rho}[P_i \mid G,I_i] - k 
\ge \mathbb{E}^{\vsigma,\rho}[P_i \mid G,I_i]  + c -k.
$$

This concludes the proof.
\end{proof}
}

\vspace{1mm}\textit{Punishment Opportunities.}
Roughly speaking, a punishment opportunity (PO) for an interaction of agent $i$ with $j$
is a later interaction between an agent $l$ and $i$ where $l$ has been informed
of a deviation of $i$ towards $j$ and thus has the opportunity to punish $i$.
This requires the existence of a temporal path (a sequence of causally influenced interactions)
from $j$ to $l$ in the evolving graph such that $i$ cannot
interfere with information forwarded from $j$ to $l$.
Formally, given $G \in \G^*$ and agent $i$, a round-$m$ $i$-edge is a pair $(j,m)$ where $(i,j)$ is an edge in $G^m$.
We say that $j$ causally influences $l$ in $G$ between $m$ and $m'$\,\cite{Kuhn:10},
denoted $(j,m) \leadsto^G (l,m')$, if $m < m'$ and either $j=l$ or there is a $j$-edge 
$(o,m'')$ in $G$ such that $(o,m'') \leadsto^G (l,m')$.
We say that $j$ causally influences $l$ in $G$ without
interference from $i$ between rounds $m$ and $m'$, denoted as $(j,m) \leadsto_i^G (l,m')$,
if the above holds for $o \ne i$. A PO of $i$ for $(j,m)$ in $G$
is an $i$-edge $(l,m')$ such that $(j,m) \leadsto_i^G (l,m')$.

\vspace{1mm}\textit{Evasive strategies.}
Given a protocol $\vsigma^*$, an evasive strategy $\sigma_i'$ for agent $i$
is a strategy where $i$ first deviates from $\sigma_i^*$ by defecting
some neighbours, and then hides this deviation from as many agents
as possible, for as long as possible.\shortv{ In particular, $i$ may defect a neighbour $j$ 
and then behave as if nothing happened,
so that the probability of neighbours of $i$ not causally influenced by $j$
observing each information set is the same, whether $i$ follows $\sigma_i^*$ or $\sigma_i'$.}

\vspace{1mm}\textit{Indistinguishable Evolving Graphs.}
We say that an evolving graph $G$ is indistinguishable from $G'$ to agent $i$ at round $m$
if $i$ acquires the same information about $G$ and $G'$, regardless of the protocol
followed by agents. Formally, given $G \in \G^*$ and round $m' \ge m$, let
$\G_j^{m'}(G)$ be the set of round-$m'$ graphs that provide the same the information to $j$ 
about the round-$m'$ topology as $G^{m'}$, and let $C_i^{m'}(G)$
be the set of agents $j$ such that $(j,m') \leadsto^G (i,m)$;
$G$ is indistinguishable from $G'$ to $i$ at $m$ iff, 
for every round $m' \le m$, $C_i^{m'}(G) = C_i^{m'}(G')$ and
$\G_j^{m'}(G) = \G_j^{m'}(G')$ for all $j \in C_i^{m'}(G)$.

\fullv{
Proposition~\ref{prop:indist} shows the intuitive result that,
if an evolving graph $G$ is indistinguishable from $G'$ to an agent $j$ at round $m'$,
then the probability of $j$ observing each round-$m$ information set under
any protocol is the same, whether the evolving graph is $G$ or $G'$,
even when conditioning on an information set of another agent $i$.
Given a run $r$, let $r_j(m')$ denote the round-$m'$ information set $I_j$
such that $r(m') \in I_j$.

\begin{proposition}
\label{prop:indist}
Given $G,G' \in \G^*$, agents $i$ and $j$, round-$m$ $I_i \in \I_i(G)$, and round-$m'$ $I_j \in \I_j(G)$,
if $G$ and $G'$ are indistinguishable to $j$ at $m'$, then for every protocol $\vsigma$ 
and belief system $\mu^*$ consistent with $\vsigma$ and $\G^*$, we have
$$\sum_{h \in I_i}\mu^*(h \mid I_i) \sum_{r: r_j(m) = I_j} \pra{r}{\vsigma}{G,h}= \sum_{h \in I_i}\mu^*(h \mid I_i) \sum_{r: r_j(m) = I_j} \pra{r}{\vsigma}{G',h}.$$
\end{proposition}
\begin{proof}
Given a protocol $\vsigma$, round $m'' \le m'$, $G'' \in \G^*$,
set $S_{m''}$ of agents $l$ such that $(l,m'') \leadsto^G (j,m')$,
and round-$m''$ information set $I_{S_{m''}} = \cap I_{l \in S_{m''}}$
fixing the observations of agents in $S_{m''}$, 
let $\pra{I_{S_1}}{\vsigma}{G''} = \sum_{h \in I_{S_{m''}}} \pra{h}{\vsigma}{G''}$.
We can show that for every round $m''$ and round-$m''$ $I_{S_{m''}}$, we have
$$\pra{I_{S_{m''}}}{\vsigma}{G} =  \pra{I_{S_{m''}}}{\vsigma}{G'}.$$
At round $1$, fix action profile $\va_{S_1}$. By definition of indistinguishable graphs,
we have that all agents in $S_{1}$ have the same information about the round-$1$ topology in $G$ and in $G'$
and the same set of neighbours.
Therefore, we have
$$\pra{I_{S_1}}{\vsigma^*}{G} = \pra{\va_{S_1}}{\vsigma^*}{G} = \pra{\va_{S_1}}{\vsigma^*}{G'} = \pra{I_{S_1}}{\vsigma^*}{G'}.$$

Now, fix $I_{S_{m''}}$. By the same reasons as above, agents in $S_{m''+1}$
follow every action profile $\va_{S_{m''+1}}$ with a probability that depends only
on $I_{S_{m''}}$ and $\vsigma^*$. Since the probability of $I_{S_{m''}}$ is the same, whether the evolving
graph is $G$ or $G'$, and since $I_{S_{m''}}$ and $\va_{S_{m''+1}}$ completely determine the 
resulting information set $I_{S_{m''+1}}$, the inductive step holds.

By consistency of beliefs and the fact that $j \in S_{m'}$,
the result follows directly from the induction.
\end{proof}
}

\section{Enforcing Accountability with Weakest Adversary}
\label{sec:necessary}
We identify a necessary restriction on $\G$ for enforcing accountability,
and provide a protocol that enforces accountability assuming only this restriction
and that pairwise exchanges are valuable.

\subsection{Need for Timely Punishments}
\label{sec:minimal}
The weakest restriction on $\G^*$ is called \emph{timely punishments}:
it says that, for some bound $\rho > 0$, and for every $G \in \G^*$, agent $i$, and $i$-edge $(j,m)$,
there must be a PO $(l,m')$ of $i$ in $G$ for $(j,m)$ such that $m' < m + \rho$.
The need for this restriction is fairly intuitive.
If the adversary is not restricted by timely punishments, then
there is $G \in \G^*$ and an agent $i$ such that either
(1) for some $i$-edge $(j,m)$, there is no PO of $i$ in $G$ for that $i$-edge,
or (2) there is no limit on the time it takes between an interaction
of $i$ and a corresponding PO. In case (1), $i$ can follow an evasive strategy to ensure that
no agent capable of punishing $i$ learns about the defection,
thus never being punished.
In case (2), $i$ can delay a punishment for an arbitrarily long time;
the problem here is the discount factor $\delta$: for an arbitrarily large constant $d$,
there is an interaction of $i$ such that, if $i$ defects the neighbour and later
follows an evasive strategy, then $i$ is only punished after $d$ rounds,
and the utility loss of this punishment is discounted by $\delta^d$;
for a sufficiently large $d$, the immediate gain of defecting outweighs the loss.
Unfortunately, some real networks do not always admit timely punishments.
For instance, in a file sharing application, agents with similar interest may exchange
files frequently, but occasionally they may interact with agents with different interests.
If agents $i$ and $j$ have different interests and $i$ happens to interact with $j$,
then $j$ may never be able to report a defection of $i$ to agents with interests similar to $i$,
which may be the only timely PO's of $i$.

Theorem~\ref{theorem:citi} shows that
timely punishments are necessary to enforce accountability.\shortv{ The proof is in Appendix~\ref{app:citi}.}
\fullv{
In the proof, an agent $i$ follows a single evasive strategy, which we now define.
In a single evasive strategy, agent $i$ deviates from a strategy $\sigma_i$ 
by defecting a neighbour $j$ and then behaves as if $i$ had not deviated,
such that the probability of neighbours of $i$ not causally influenced by $j$
observing each information set is the same, whether $i$ follows $\sigma_i'$ or $\sigma_i$.

Formally, we define a single evasive strategy $\sigma_i'$ inductively relative 
to a protocol $\vsigma$.
Let $\vsigma' = (\sigma_i',\vsigma_{-i})$.
For every round $m' < m$, $\sigma_i'$ is identical to $\sigma_i$.
At every round-$m$ $I_i \in \I_i(G)$,
with probability $\sigma_i^*(a_i^* \mid G^1)$,
$i$ follows $a_i$ identical to $a_i^*$ except $i$ defects $j$.
After observing round-$m+1$ $I_i'$ consistent
with $i$ following $a_i$, $i$ selects $I_i''$ identical to $I_i'$
except $I_i''$ is compatible with $i$ following $a_i^*$ at round $m$.
Given $m' \ge m$, round-$m'$ $I_i \in \I_i(G)$,
and corresponding $I_i' \in \I_i(G)$ selected by $i$,
$i$ follows $a_i$ with probability $\sigma_i(a_i \mid I_i')$.
After observing round-$m'+1$ $I_i''$ consistent with $i$ following $a_i'$,
$i$ selects $I_i'''$ compatible with $I_i''$ (same observations by $i$ prior to round $m'$),
and compatible with the observations made by $i$ in round-$m'$ in $I_i''$.
We show that, if there is no PO of $i$ for $(j,m)$ in $G$ up to round $m'$,
then every neighbour of $i$ in rounds between $m$ and $m'$ makes observations with the same probability,
whether $i$ follows $\sigma_i'$ or $\sigma_i$. Notice that now runs of $\vsigma''$
also specify the choices of information sets made by $i$ prior to every round.

Fix round-$m$ $I_i \in \I_i(G)$. Given history $h \in I_i$,
round $m' \ge m$, set $S_{m'}$ of agents not causally influenced by $j$ 
between $m$ and $m'$ in $G$, and round-$m'$ information set $I_{S_{m'} \cup i}^* = \cap_{l \in S_{m'} \cup i} I_l^*$ fixing the
observations of agents in $S_{m'}$,
let $Q^{\vsigma'}(I_{S_{m'} \cup \{i\}} \mid G,h)$ be the probability of agents in $S_{m'}$ observing $I_{S_{m'}}^*$
and $i$ selecting $I_{i}^*$. Lemma~\ref{lemma:hide} shows that $Q^{\vsigma'}(I_{S_{m'} \cup \{i\}}^* \mid G,h)$
is the same as the probability of agents in $S_{m'} \cup \{i\}$ observing $I_{S_{m'} \cup \{i\}}$ when they follow $\vsigma$.
Since there is no PO of $i$ for $(j,m)$ up to round-$m'$, every neighbour of $i$ is in $S_{m'}$. 
This implies that every neighbour of $i$ makes the same 
observations between $m$ and $m'$ with the same probability, whether 
$i$ follows $\sigma_i'$ or $\sigma_i$.

\begin{lemma}
\label{lemma:hide}
If there is no PO of $i$ in $G$ for $(j,m)$ up to round $m'$, then for every $h \in I_i$, round $m'' > m$ with $m'' \le m'$,
and round-$m''$ information set $I_{S_{m''} \cup \{i\}}$, we have
\begin{equation}
\label{eq:hide}
Q^{\vsigma'}(I_{S_{m''} \cup \{i\}} \mid G, h) = \sum_{h' \in I_{S}} \pra{h'}{\vsigma}{G,h}.
\end{equation}
\end{lemma}
\begin{proof}
First, consider round $m'' = m+1$ and fix $h$.
Clearly, every agent different from $i$ follows each round-$m$ action with the same probability,
whether $i$ follows $\sigma_i'$ or $\sigma_i$. In addition,
when following $\sigma_i'$, $i$ selects $a_i^*$ with probability $\sigma_i(a_i^* \mid I_i)$,
and sends messages according to $a_i^*$ except to $j$.
Since $S_{m+1}$ includes all agents but $j$ and $i$, then every $l \in S_{m+1}$
makes each observation with the same probability whether $i$ follows $\sigma_i$ or $\sigma_i'$,
and $i$ selects each action with the same probability as if $i$ were following $\sigma_i$.

Now, suppose that the hypothesis holds for $m'' > m$ with $m'' < m'$.
Fix round-$m''$ $I_{S_{m''} \cup \{i\}}^1$.
Since $i$ has no PO for $(j,m)$ up to $m''$, for every agent $l \in S_{m''+1} \cup \{i\}$,
$l \in S_{m''} \cup \{i\}$, so $l$ follows each action $a_l^*$ with probability $\sigma_l(a_i^* \mid G, I_l^1)$.
Moreover, every round-$m''$ neighbour of $l$ is in $S_{m''} \cup \{i\}$,
hence, given $I_{S_{m''} \cup \{i\}}$, $l$ makes observations in round $m''$ (if $l=i$, $i$ selects each information set)
with the same probability, whether $i$ follows $\sigma_i$ or $\sigma_i'$.
The result follows directly from the hypothesis.
\end{proof}

We can now prove Theorem~\ref{theorem:citi}.
}

\begin{restatable}{theorem}{theoremciti}
\label{theorem:citi}
If the adversary is not restricted by timely punishments, then there is no protocol that enforces accountability.
\end{restatable}
\fullv{
\begin{proof}
The proof is by contradiction. Suppose that $\vsigma^*$ enforces accountability.
Suppose also that the adversary is not restricted by timely punishments.
For every $\rho > 0$, we can fix $G \in \G^*$, agent $i$, and $i$-edge $(j,m)$
such that, for every PO $(l,m')$ of $i$ for $(j,m)$ in $G$, we have $m' - m \ge \rho$.
This implies that, for every round $m' \ge m$ with $m' < m+\rho$,
and every $i$-edge $(l,m')$ with $l \ne j$, $(j,m) \leadsto^G (l,m')$ is false, since the
existence of a PO $(l,m')$ of $i$ for $(j,m)$ in $G$ with $m' < m+\rho$ is equivalent
to the existence of one such $(l,m')$ with $(j,m) \leadsto^G (l,m')$.
Let $\rho$ be such that $y \delta^{\rho}/(1-\delta) < 1$,
where $y$ is the (bounded) maximum difference between utilities of a single interaction.
Since $\sigma_i^*$ enforces accountability, $i$ is expected to cooperate at 
every round-$m$ information set $I_i \in \I_i(G)$ consistent with $\vsigma^*$ and $G$.
Here, $i$ may follow a single evasive strategy $\sigma_i'$ 
relative to $\vsigma^*$, $G$, and $(j,m)$ \shortv{(See Appendix~\ref{app:evasive})}. 
Let $\vsigma' = (\sigma_i',\vsigma^*_{-i})$.
By Lemma~\ref{lemma:hide}, for every round $m' < m+\rho$ with $m' \ge m$,
the round-$m'$ neighbours of $i$ observe each round-$m'+1$ information
set with the same probability, whether $i$ follows $\sigma_i'$ or $\sigma_i^*$.
Since these information sets specify the individual actions followed by and taken towards $i$,
the expected utility of $i$ in round $m'$ is the same,
whether agents follow $\vsigma'$ or $\vsigma^*$,
except $i$ avoids at least the cost $1$ of defecting $j$ in round $m$.
Moreover, the maximum utility difference in every round $m' \ge m+\rho$ is $y n$.
Therefore, we have
$$u_i(\vsigma^* \mid G,I_i) - u_i(\vsigma' \mid G,I_i) < -1 + \delta^{\rho} yn/(1-\delta) < 0.$$

This contradicts the assumption that $\vsigma^*$ enforces accountability, concluding the proof.
\end{proof}
}

\subsection{A \oape{\G^*} for Valuable Pairwise Exchanges with Timely Punishments}
\label{sec:valuable}
We now describe a safe-bounded protocol $\vsigma^{\val}$ that enforces accountability
in a setting of valuable pairwise exchanges (defined below),
assuming an adversary restricted by timely punishments.
In valuable pairwise exchanges, we assume that agents can perform proportional punishments
and punishment avoidance actions. This implies that agents can punish each other
with a cost proportional to $\pi$, and they can avoid such punishments at the expense
of not receiving the value from their neighbours.
We also assume that the benefit-to-cost ratio of receiving/sending a value is high.
Specifically, we assume that $\beta > 1 + \alpha + \rho\pi$ and $\pi > n$,
where $\rho$ is the constant in the definition of timely punishments.
\shortv{In Appendix~\ref{app:penances},}\fullv{In Section~\ref{sec:valuable},}
we discuss exchanges that can be modelled in this way,
assuming that agents exchange multiple messages per round
and neglect download costs ($\alpha = 0$).

We now describe $\vsigma^{\val}$.
In every round, agents exchange values and monitoring information that
includes accusations revealing defections. These accusations are disseminated across
the network, and used by agents to adjust the cost of a punishment applied to neighbours.
More precisely, let $\rho$ be as in the definition of timely punishments.
At the beginning of round $m$, agent $i$ keeps for every agent $l$ and
round $m' \in \{m-\rho \ldots m-1\}$ a report indicating whether
$l$ defected some neighbour in round $m'$. We call a report
indicating a defection an accusation.
In a round-$m$ action, agents $i$ and $j$ exchange
their reports and values, and apply punishments proportional
to the number of accusations against each other.
At the end of round $m$, $i$ updates its reports relative
to every $l \ne j$ and $m' < m$ if $j$ does not defect $i$,
otherwise $i$ emits an accusation against $j$ for round $m$.

Theorem~\ref{theorem:penance} shows that $\vsigma^{\val}$ enforces accountability
in valuable pairwise exchanges\shortv{ (see Appendix~\ref{app:penances})}.
We make two minimal assumptions: (1) the adversary is restricted by timely punishments,
which as we have seen is strictly necessary, and (2) agents are \emph{sufficiently patient},
i.e., the factor $\delta$ is sufficiently close to $1$, which is a standard assumption in proofs
of Folk Theorems, and is necessary for future losses of punishments to always outweigh
the gains of deviating in the present (recall that future losses are discounted to the present by $\delta$).
The proof shows that a defection of $i$ is always matched by a punishment that occurs
after at most $\rho$ rounds. $i$ gains at most $n$ by defecting
neighbours in a round, while losing at least $\delta^{\rho} \pi$.
Given that $\beta > 1+\alpha + \rho\pi$ and $\pi > n$,
if $\delta$ is sufficiently close to $1$, then the loss outweighs the gain.
Moreover, if $i$ avoids the cost of a punishment and of sending and receiving messages (at most $1 + \alpha + \rho \pi$),
$i$ does not receive the value, thus losing $\beta$.
Since $i$ can never influence the reports relative to itself, 
$i$ never gains by deviating.

\begin{restatable}{theorem}{theorempenance}
\label{theorem:penance}%
If the adversary is restricted by timely punishments and agents are sufficiently patient,
then $\vsigma^{\val}$ enforces accountability in valuable pairwise exchanges.
\end{restatable}
\fullv{
\begin{proof}
We show that $\vsigma^{\val}$ is a \oape{\G^*}.
Fix $G \in \G^*$, agent $i$, round-$m$, round-$m$ information set $I_i \in \I_i(G)$, $h \in I_i$,
and actions $a_i^*,a_i'$ such that $\sigma_i^{\val}(a_i^*|I_i) > 0$.
Fix $G \in \G^*$, agent $i$, round-$m$, round-$m$ information set $I_i \in \I_i(G)$, $h \in I_i$,
and actions $a_i^*,a_i'$ such that $\sigma_i^{\gen}(a_i^*|I_i) > 0$.
Let $\Delta = \eu{i}{\vsigma^1}{h,G} - \eu{i}{\vsigma^2}{h,G}$,
where $\vsigma^1 = \vsigma^{\gen}|_{I_i,a_i^*}$, $\vsigma^2 = \vsigma^{\gen}|_{I_i,a_i'}$,
and $\vsigma^{\gen}|_{I_i,a_i}$ represents the protocol that differs from $\vsigma^{\gen}$
exactly in that $i$ deterministically follows $a_i$ at $I_i$.
\shortv{In the full paper,
we show that the result known
as One-Shot-Deviation principle\,\cite{Hendon:96} applies to the solution concept \oape{\G^*}.
This principle states that}\fullv{By the One-Shot-Deviation principle,} it suffices to show that $\Delta \geq 0$ to prove the result.
For every $m' \ge m$, let $\Delta^{m'}$
be the difference between the expected utility of $i$ in round $m'$,
between agents following $\vsigma^1$ and $\vsigma^2$, i.e.,
$$\Delta^{m'}  =  \sum_{r} \pra{r}{\vsigma^1}{G,h} u_i(\va^{m'}) - 
 \sum_{r} \pra{r}{\vsigma^1}{G,h} u_i(\va^{m'}),$$
where $\va^{m'}$ is the round-$m'$ action profile in $r$.
We show the following facts, where again we say that
a fact holds after $i$ follows $\va_i'$ or $\va_i^*$
if the fact holds for all runs $r^1$ and $r^2$ in $G$ of $\vsigma^1$ and $\vsigma^2$, respectively:

\begin{enumerate}[(a)]
  \item \emph{Fact 1}: If $m'\ge m$ and $i$ follows $a_i'$, then $j$ has an accusation 
  against $i$ for round $m$ at the end of $m'$ iff $i$ defects a round-$m$ neighbour $l$ and $(l,m) \leadsto_i^G(j,m'+1)$.
  \begin{proof}
  At the end of round $m$, an agent $j$ has an accusation against $i$ iff $i$ defects $j$.
  The set of neighbours that hold accusations are exactly the set of agents $j$ such that $(l,m) \leadsto_i^G (j,m+1)$,
  where $l$ is a round-$m$ neighbour that $i$ defects.  
  Continuing inductively, at the end of every subsequent round $m'$, an agent $j$ has an accusation
  if $j$ already had that accusation or received it from $l \neq i$. Either way, by the induction hypothesis,
  this is only the case if $i$ defected some round-$m$ neighbour $l$ and $(l,m) \leadsto_i^G (j,m'+1)$.
    \end{proof}
   
  \item \emph{Fact 2}: $j$ has an accusation against $i$ for $m'' \neq m$ at the end of round $m' \ge m$
  after $i$ follows $a_i^*$ iff the same is true after $i$ follows $a_i'$:
  \begin{proof}
  When $m' = m$, agents ignore  the accusations relative to $i$ and $m''$ sent by $i$ in round $m$,
  so the actions $a_i^*$ and $a_i'$ have the same impact on the accusations held by other agents after $m'$.
  The same is true in every subsequent round $m' > m$, except when $m' = m''$. In this case, since
  $i$ follows $\sigma_i^{\val}$ in round $m'$, both after following $a_i^*$ and $a_i'$ in round $m$,
  $i$ defects no neighbour and thus no agent will have an accusation against $i$ relative to $m''$ at the end of $m'$.
  \end{proof}
  
  \item \emph{Fact 3}: If $m' > m$, then $\Delta^{m'} \ge 0$.
  \begin{proof}
  Since $i$ follows $\sigma_i^{\val}$ after $m$, in round $m'$, $i$ either cooperates with or punishes every
neighbour, and these actions have a fixed cost $1$, whether $i$ follows $a_i^*$ or $a_i'$.
Moreover, $i$ always receives the values from all neighbours, obtaining a benefit $\beta$
and incurring a fixed cost $\alpha$. So, the utility in round $m'$ of $i$ following $a_i'$ may differ from that of $i$ following $a_i^*$
only if the neighbours of $i$ have a different number of accusations against $i$.
By Facts~1 and 2, the number of accusations when $i$ follows $a_i'$ is at least as high
as when $i$ follows $a_i^*$. This implies that $\Delta^{m'} \ge 0$.
  \end{proof}
  
  \item \emph{Fact 4}: If $i$ defects a neighbour in $a_i'$, then there exists a round $m' < m+\rho$
  such that $\Delta^{m'} \ge \pi$.
  \begin{proof}
  Suppose that $i$ defects a neighbour $j$ in $a_i'$.
 There is a PO $(l,m')$ of $i$ for $(j,m)$ in $G$ with $m' - m < \rho$,
  such that $(j,m) \leadsto_i^G (l,m')$. By Fact~1, $l$ has an accusation 
  against $i$ for round $m$ when $i$ follows $a_i'$
  and no accusation relative to $m$ when $i$ follows $a_i^*$.
  As in the proof of Fact~3, $i$ incurs the same costs of cooperating or punishing neighbours and obtains the same 
  benefit and incurs the same cost of receiving values, whereas $i$ incurs the additional cost $\pi$ of being punished by $l$
  when $i$ follows  $a_i'$. Therefore, $\Delta^{m'} \ge \pi$.
  \end{proof}
\end{enumerate}

We conclude by showing that $\Delta \ge 0$.
Suppose that in $a_i'$ agent $i$ avoids punishments from $x$ neighbours
and defects $y$ neighbours. By avoiding a punishment from $j$, $i$
avoids at most the cost $1 + \alpha +  \rho \pi$, but loses the benefit $\beta$,
thus the variation in the utility of these interactions is $\beta - 1- \alpha - \rho\pi > 0$.
By defecting $j$, $i$ gains at most $1$. This implies that $\Delta^m \ge -y$.
By Facts~3 and 4, we have
$$\Delta \ge \delta^{\rho} \pi - y \ge \delta^{\rho}\pi - n.$$
By $\pi > n$, if $\delta$ is sufficiently close to $1$, $\Delta \ge 0$,
as we intended to prove.
\end{proof}

\subsection{Examples of Valuable Pairwise Exchanges}
We now provide examples of interactions that meet the restrictions of valuable pairwise exchanges.
First, we consider an interaction without cryptography. Then, we show how cryptography
can decrease the number of restrictions on the utility.

\paragraph{Without cryptography.} Consider that every neighbouring 
agents $i$ and $j$ can exchange messages in three phases per round
and neglect download costs, which may be the case if bandwidth is asymmetric.
In phase 1, $i$ send reports indicating the number $c_j$ of accusations against $j$, and similarly
$j$ may send the equivalent number $c_i$ of accusations against $i$. In phase 2, $i$ and
$j$ send $c_i$ and $c_j$ penance messages that cost $\pi$ each, respectively.
In phase 3, they exchange their values only if both agents have sent all required information in previous phases.
In this setting, agents can adjust the size of penances such that $\pi > n$.
Suppose that the cost of sending phase 1, 2, and 3 messages is $1$.
If $i$ sends all requested information in the first two phases, $i$
does not avoid the punishment of sending the $c_i$ penance messages,
while avoiding at most the cost $1$ of defecting $j$ by not sending the value in phase 3.
As shown by Theorem~\ref{theorem:penance}, $i$ does not gain from defecting.
$i$ can avoid the punishment by not sending the requested messages in the first two phases,
but in this case $i$ does not receive the value from $j$ in phase $3$,
also not gaining from this deviation.

\paragraph{With cryptography.} The previous interactions require knowledge of $\beta$ 
in order to appropriately define the size of a penance message,
and require $\beta = \Omega(n\rho)$, making it restrictive in practice if $n$ is large.
We can mitigate these problems using the technique of delaying gratification from~\cite{Li:06}.
Instead of exchanging values in phase 3 of round $r$, neighbours $i$ and $j$ 
exchange the values ciphered with private keys in phase $1$ along with monitoring information.
In phase 2, they still send the penances and in phase 3 they exchange the keys. 
As before, if messages are omitted in phase 1 or 2, then the agents send no further messages.
Let $\alpha^\kappa$ be the cost of sending the key.
This mechanism only requires  $\pi > \alpha^k$ and $\beta > 1 + n\rho\pi$.
Thus, it suffices that $\beta > 1  +  \epsilon$, where $\epsilon = n\rho\alpha^\kappa$.
If $\alpha^\kappa \ll 1$, then $1+ \epsilon$ is close to the optimal restriction of $\beta > 1$.
It is also possible to adjust $\pi$ without knowing $\beta$:
we can define the size of a penance to be larger than that of a key, but smaller than the value.
The arguments that show that this mechanism is a \oape{\G^*} are the same as Theorem~\ref{theorem:penance}.
}

\section{Enforcing Accountability in General Pairwise Exchanges}
\label{sec:sufficient}
We now address the problem of enforcing accountability without making the 
assumption that pairwise exchanges are valuable. We consider the least restrictive
assumptions about utility and individual actions available to agents.
First, we consider the smallest benefit-to-cost ratio of receiving/sending values.
We still need the benefit $\beta$ to be larger than the total costs $1+ \alpha$ of sending and receiving messages,
or else agents would not have incentives to engage in exchanges.
Second, we do not assume that agents neglect download costs ($\alpha \ge 0$),
nor that they communicate in multiple steps. Therefore, proportional punishments
and punishment avoidance are not available actions. We still assume
that agents can defect, cooperate, or (actively) punish other agents.
Since there is a trivial one-shot implementation of active punishments,
where an agent sends garbage of the size of the value instead of the value,
we consider that a punishment causes a utility loss of $\pi \ge \beta$.

With this in mind, we say that a protocol enforces accountability in general
pairwise exchanges if it is an equilibrium for all utilities such that $\beta > 1+\alpha$,
and only has agents following the three aforementioned actions.
\shortv{We show in the full paper}\fullv{First, we show}
that non-safe protocols cannot in general enforce accountability.
\shortv{In this section,}\fullv{Then,} we identify a necessary and
a sufficient condition for enforcing accountability with safe-bounded protocols
in general pairwise exchanges. The results are generalized for non-bounded protocols
in the full paper.

\fullv{
\subsection{Need for Safe Strategies}
\label{sec:problems}
Existing proofs of Folk Theorems devise protocols where agents
provide the minimax utility in punishments\,\cite{Mailath:07}.
In pairwise exchanges, the minimax utility is $0$, and is obtained when neighbouring agents defect each other.
With such strategies, deviations during punishments provide no gain and hence
require no further punishments.
We identify two scenarios that illustrate problems of applying minimax punishments in dynamic networks,
which occur in both types of exchanges considered in this paper.
The first scenario occurs when an agent $i$ deviates and in a later interaction with $j$
agent $i$ does not know whether $j$ will punish $i$ (in which case $i$ is allowed to defect)
or expects $i$ not to defect and triggers additional punishments if $i$ defects.
In the second scenario, an agent $j$ knows that a defection between $l$ and $i$ will occur,
and $j$ knows that if he defects another agent, then any monitoring information triggering
future punishments will be lost once $l$ defects $i$.\shortv{
Due to lack of space, we defer the second scenario to Appendix~\ref{app:unsafe}.}
Although these scenarios do not imply that we can never use non-safe strategies,
they suggest that similar problems may arise in general,
thus non-safe strategies may require complicated mechanisms 
for distinguishing between exchanges where there may and may not be defections.
Since it is not clear that such distinction can be made with bounded memory
(agents may need to distinguish between evolving graphs),
we focus on safe strategies in the remainder of the paper.

\begin{figure}
\begin{center}
 \includegraphics[scale=0.2]{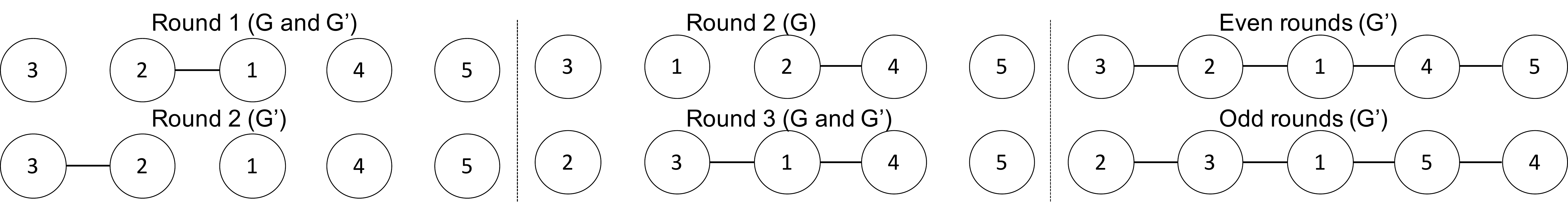}
\end{center}
\vspace{-0.2in}
\caption{Ambiguous PO - $G$ and $G'$ are indistinguishable to 1 at round $3$.}
\label{fig:ambig}
\end{figure}

\subsubsection{Ambiguous Punishment Opportunities}
Consider the scenario depicted in Fig.~\ref{fig:ambig}.
There are five agents numbered 1 to 5,
and two evolving graphs $G$ and $G'$ from $\G^*$.
In round~$1$, agent~1 interacts with agent~2, both in $G$ and $G'$.
In round~$2$, agent~2 interacts with agent~4 in $G$,
whereas 2 interacts with agent~3 in $G'$.
Finally, in round~$3$ agent~1 interacts with agents~3 and~4 both in $G$ and $G'$.
Suppose that, if the adversary selects $G$, then 4 is the only agent capable
of punishing a defection of 1 towards 2 in a timely fashion,
and consider that after round~$3$, $G'$ forms a line
where the edges to agent~1 always form a cut between $\agents_1 = \{2,3\}$ and $\agents_2=\{4,5\}$,
and the set of neighbours of 1 alternates between $\{3,4\}$ and $\{2,5\}$.
Notice that $G'$ admits timely PO's for every agent $i$ and $i$-edge.
If the only information about the graph topology available to agents
in each round is the set of their neighbours, then 1 cannot distinguish
between $G$ and $G'$ at round $3$.
Suppose that there is a protocol $\vsigma^*$ that enforces
accountability and requires a mutual defection between 1 and 4 if 1 defects 2.
In $G$, since $\vsigma^*$ is a \oape{\G^*}, 1 obtains a higher utility for defecting 4.
Now, consider that the adversary selects $G'$, and consider the two cases where
(a) 1 defects 2 in round $2$ and defects 4 because 1 is following $\sigma_1^*$,
and (b) 1 does not deviate from $\sigma_1^*$ in round $2$ and deviates in round $4$ by defecting 4.
Agent 1 can devise an evasive strategy where 
1 cooperates with 4 and 5 and hides the defection towards 2 from them,
such that agents 4 and 5 never punish 1. On the other hand, if $i$ follows $\sigma_i^*$
and defects 4 in round~$3$, then either 4 or 5 must later punish 1, because
they cannot distinguish between cases (a) and (b), and in case (b)
they are the only agents capable of punishing 1 for the defection.
This implies that 1 obtains a strictly higher utility of cooperating with 4
rather than following $\sigma_i^*$ and defecting 4,
contradicting the assumption that $\vsigma^*$ enforces accountability.

Formally, given agents $i,j$ and $G,G' \in \G^*$,
we say that there is a round-$m$ partition of $G'$ regarding $i$ and $j$ if 
there is a partition $(\agents_1,\agents_2)$ of $\agents \setminus \{i\}$ satisfying
(1) $j \in \agents_2$, (2) for every two $i$-edges $(l,m')$ and $(o,m'')$ in $G'$,
$(l,m') \leadsto_i^{G'} (o,m'')$ only if $l$ and $o$ are in the same set of the partition,
and (3) for every $i$-edge $(l,m')$ with $m' < m$, we have $l \in \agents_1$.
We say that $i$-edge $(j,m)$ in $G \in \G^*$ is an ambiguous PO
iff there is $G' \in \G^*$ such that $G'$ is indistinguishable to $i$ from $G$ at $m$ and
there is a round-$m$ partition $(\agents_1,\agents_2)$ of $G'$ regarding $i$ and $j$.
We say that $\vsigma^*$ has ambiguous defections iff there is $G \in \G^*$, agent $i$,
and ambiguous $i$-edge $(j,m)$ in $G$ such that $\sigma_i^*$ requires $i$ to defect $j$
at some information set $I_i \in \I_i(G)$ consistent with only $i$ deviating from $\vsigma^*$,
i.e., there is $\sigma_i' \ne \sigma_i^*$ such that $I_i$ is consistent with $(\sigma_i',\vsigma^*_{-i})$ and $G$.

Theorem~\ref{theorem:ambiguous} proves that a $\vsigma^*$ that has ambiguous defections
cannot be a \oape{\G^*}. The proof uses the notion of dual evasive strategy,
which we now define. In a dual evasive strategy $\sigma_i''$, $i$ cooperates with an agent $j$ in round $m$
where $(j,m)$ is an ambiguous PO\shortv{ (see Appendix~\ref{app:ambig})},
and tries to convince agents from the two disjoint 
sets $\agents_1$ and $\agents_2$ that $i$ is following two different
strategies $\sigma_i'$ and $\sigma_i^*$, respectively,
where, in $\sigma_i'$, $i$ defects some agents from $\agents_1$
and follows $\sigma_i^*$ afterwards.
This way, $i$ hides the defections from the agents in $\agents_2$,
while not revealing to agents in $\agents_1$ that $i$ is hiding
the defections from agents in $\agents_2$.
The idea of $\sigma_i''$ is similar to a single evasive strategy:
after round $m$, $i$ sends messages to agents in $\agents_2$
as if $i$ did not defect any agent, and sends messages to agents in $\agents_1$
as if $i$ were following $\sigma_i''$.

Specifically, suppose that $i$-edge $(j,m)$ is an ambiguous PO,
and fix corresponding $G$, $G'$, $\vsigma^*$, $\vsigma'$, and $I_i$,
such that $\vsigma^*$ has an ambiguous defection, $I_i$ is a round-$m$ information
set consistent with $\vsigma'$ and $G$, and in $\vsigma'$ only deviates from $\sigma_i^*$,
and only deviates before round $m$.
We define a dual evasive strategy $\sigma_i''$ as follows.
For every round $m' < m$ or information set $I_i'$ incompatible with $I_i$,
$\vsigma''$ is identical to $\vsigma^*$.
The definition of $\sigma_i''$ for the remaining information sets is inductive.
At every round $m' \ge m$, $\sigma_i''$ specifies a probability
distribution over a pair of round-$m$ histories $(h^1,h^2)$ 
and round-$m$ actions followed by $i$. Intuitively, $h^1$ is used to emulate
the behaviour of agents in $\agents_1$ and $h^2$ is used to emulate the behaviour of agents in $\agents_2$.

Consider round $m$. Since $G'$ is indistinguishable from $G$ at $m$,
we have $I_i \in \I_i(G')$ and $I_i$ is consistent with $\vsigma'$ and $G'$.
Agent $i$ selects a round-$m$ history $h^1 \in I_i$ 
with probability $\pra{h^1}{\vsigma^*}{G'}/\sum_{h' \in I_i'} \pra{h'}{\vsigma^*}{G}$,
where $I_i'$ is some round-$m$ information set consistent with $\vsigma^*$ and $G$.
In addition, $i$ selects $h^2$ with probability $\pra{h^2}{\vsigma'}{G'}/\sum_{h' \in I_i} \pra{h'}{\vsigma^*}{G'}$.
Now, consider round $m' \ge m$, and suppose that $\vsigma''$ is defined for every previous round.
Fix a pair $(h^1,h^2)$ selected at the end of $m'-1$.
At round $m'$, agent $i$ does the following.
First, $i$ selects two action profiles $\va^1_{\agents_2 \cup \{i\}}$ and $\va^2_{\agents_1 \cup \{i\}}$
with probabilities $\pra{\va^1_{\agents_2 \cup \{i\}}}{\vsigma^*}{G',h^2}$
and $\pra{\va^2_{\agents_1 \cup \{i\}}}{\vsigma^*}{G',h^1}$, respectively.
These are the actions that agents in $\agents_2$ and $\agents_1$ could have followed
had $i$ followed $\sigma_i^*$ at $I_i$ and $\sigma_i^*$ from the beginning, respectively.
Then, $i$ follows individual actions towards agents in $\agents_1$ according
to $a_i^1$ and follows individual actions towards agents in $\agents_2$ according to $a_i^2$.
After observing round-$m'+1$ information set $I_i'$,
$i$ deterministically selects two action profiles $\va^1_{\agents_1}$ and $\va^2_{\agents_2}$
such that the round-$m'$ individual actions followed by agents in $\agents_1$ towards $i$ in $I_i'$
are the same as those followed by these agents in $\va^1_{\agents_1}$, and the 
same holds for agents in $\agents_2$ and $\va^2_{\agents_2}$.
Let $\va^1 = (\va^1_{\agents_1},\va_{\agents_2 \cup \{i\}}^1,\va_{\agents_2}^1)$
and $\va^2 = (\va_{\agents_1 \cup \{i\}}^2,\va_{\agents_2}^2)$.
The selected round-$m'+1$ histories are $(h^1,\va^2)$ and $(h^2,\va^1)$,
where $(h',\va')$ represents the history that results from appending $\va'$
to the sequence of action profiles in $h'$. This concludes the definition.

We prove Lemma~\ref{lemma:dual},
which shows that by following $\vsigma''$, agent $i$ leads
agents in $\agents_1$ to believe that $i$ follows $\sigma_i^*$ at $I_i$,
and lead agents in $\agents_2$ to believe that all agents followed $\vsigma^*$ since the beginning.
The proof is similar to that of Lemma~\ref{lemma:hide}.

Given round-$m'$ history $h$, let $Q^{\vsigma''}_1(h \mid G',I_i)$ be the probability of
$h$ specifying the actions followed by agents in $\agents_2$ conditional on $I_i$,
and the actions of other agents in $h$ being the same as in the round-$m$ history $h^1$ selected by $i$.
Let $Q^{\vsigma''}_2(h \mid G',I_i)$ denote the corresponding probability regarding the
actions of agents in $\agents_1$ and selection of $h^2$.
Fix $\mu^*$ consistent with $\vsigma^*$ and $\G^*$.
Lemma~\ref{lemma:dual} shows that $Q^{\vsigma''}_1(h \mid G',I_i)$ is the same as 
the belief held by $i$ that the round-$m'$ history will be $h$ conditional
on $G'$, some information set $I_i'$ consistent with $\vsigma^*$ and $G'$,
and agents following $\vsigma^*$. Since $h$ determines the observations
of agents in $\agents_2$, this implies that $i$ believes that these agents will not detect
the deviations of $i$. Similarly, $Q^{\vsigma''}_2(h \mid G',I_i)$ is the same 
as the belief held by $i$ that  the round-$m'$ history will be $h$ conditional
on $G'$, $I_i$, and agents following $\vsigma^*$, hence $i$ believes
that agents in $\agents_1$ will not detect a deviation of $i$ from $\vsigma^*$ at and after $I_i$.

\begin{lemma}
\label{lemma:dual}
For every round $m' \ge m$ and round-$m'$ $h \in \hist$, there is a round-$m$ information set $I_i'$
consistent with $\vsigma^*$ and $G'$ such that
$$Q^{\vsigma''}_2(h \mid G',I_i) = \sum_{h' \in I_i} \mu^*(h' \mid G',I_i) \pra{h}{\vsigma'}{G',h'},$$
$$Q^{\vsigma''}_1(h \mid G',I_i) = \sum_{h' \in I_i'} \mu^*(h' \mid G',I_i') \pra{h}{\vsigma^*}{G',h'}.$$
\end{lemma}
\begin{proof}
We first show for $Q^{\vsigma''}_2(h \mid G',I_i)$, using induction on $m'$.

Consider $m' = m$. By the definition of ambiguous PO $(j,m)$, 
for every $l \in \agents_1$, $o \in \agents_2$, and $m' < m$,
it is false that $(l,m') \leadsto^{G'} (o,m)$ and $(o,m') \leadsto^{G'} (l,m)$.
Since $i$ does not interact with agents in $\agents_2$, the actions of
agents in $\agents_2$ do not depend on the actions of agents in $\agents_1 \cup \{i\}$, and
their probability is given by $\vsigma^*$.
Similarly, the actions of agents in $\agents_1$ 
do not depend on the actions of agents in $\agents_2$, and their probability is given by $\vsigma'$.
This implies that the belief held by $i$ conditional on $G'$ and $I_i$
that the actions of agents in $\agents_2$ are given by $h$
is the same as the probability of these agents following those actions,
conditional on $G'$, $I_i$, and agents following $\vsigma^*$;
by construction of $\sigma_i''$, this is exactly the probability of
$i$ selecting those actions. Moreover, the actions of $i$ selected by $i$ in $h$ are fixed by $I_i$;
clearly, these are the same as the actions followed by $i$ conditional on $G'$ and $I_i$. 
Finally, \fullv{by Proposition~\ref{prop:onlydev},}\shortv{ we show in the full paper that
$i$ believes at $I_i$ that only $i$ deviated, hence}
the belief that agents in $\agents_1$ follow the actions
specified in $h$ conditional on $G'$ and $I_i$ is exactly the probability that those agents
followed those actions, conditional on $G'$, $I_i$, and agents following $\vsigma'$.
Therefore, we have $\mu^*(h \mid G',I_i) = Q^{\vsigma''}_2(h \mid G',I_i)$, as we intended to prove.

Now, assume the hypothesis for round $m'$ and consider round-$m'+1$ history $(h,\va)$.
We have 
$$
Q^{\vsigma''}_2((h,\va) \mid G',I_i) = Q^{\vsigma''}_2(h \mid G',I_i) \pra{\va}{\vsigma^*}{G',h},
$$
which directly implies the result by the hypothesis.
To see this, notice that, if $h$ represents the actions taken by agents in $\agents_1$
and the selection of $h^2$, then $h_{\agents_1 \cup \{i\}}$
specifies the observations of agents in $\agents_1$,
so they follow $\va_{\agents_1}$ with probability $\pra{\va_{\agents_1}}{\vsigma^*}{h_{\agents_1 \cup \{i\}}}$,
i.e., the product of the probabilities $\sigma_l^*(a_l \mid I_l)$ for $l \in \agents_1$,
where $I_l$ is the round-$m'$ information set, which, by the definition of ambiguous PO,
is determined by the actions $h_{\agents_1 \cup \{i\}}$ of agents in $\agents_1 \cup \{i\}$.
Conversely, $i$ selects $\va_{\agents_2 \cup \{i\}}$ with probability
$\pra{\va_{\agents_2 \cup \{i\}}}{\vsigma^*}{G',h_{\agents_2 \cup \{i\}}}$.

The proof for $Q^{\vsigma''}_1(h \mid G',I_i)$ is identical,
except we revert the roles of $\agents_1$ and $\agents_2$, and $i$ initially
selects $I_l'$ consistent with $\vsigma^*$ and $G'$.
This concludes the proof.
\end{proof}

We can now prove Lemma~\ref{theorem:ambiguous}.

\begin{restatable}{theorem}{theoremambig}
\label{theorem:ambiguous}
If $\vsigma^*$ has ambiguous defections, then $\vsigma^*$ does not enforce accountability.
\end{restatable}
\begin{proof}
Fix a belief system $\mu^*$ consistent with $\vsigma^*$ and $\G^*$.
Fix $G$, $i$, and an ambiguous $(j,m)$.
Suppose that $\sigma_i^*(a_i^D \mid I_i) >0 $, where $I_i \in \I_i(G)$
is a round-$m$ information set consistent with $G$ and a
protocol $\vsigma' = (\sigma_i',\vsigma^*_{-i})$,
$a_i^D \in \act_i(G^m)$ is an action where $i$ omits a message to $j$,
and $\sigma_i'$ is some strategy identical to $\sigma_i^*$ in every round $m' \ge m$.
Consider $G'$ indistinguishable from $G$ to $i$ at $m$
as in the definition of ambiguous PO $(j,m)$,
and corresponding partition $(\agents_1,\agents_2)$.
Since the $i$-edges form a cut between agents in $\agents_1$
and in $\agents_2$ in $G'$, $i$ can control all the information
that flows between agents from the two sets.
More precisely, $i$ can follow a dual evasive strategy $\sigma_i''$\shortv{ (see Appendix~\ref{app:dual}),
where $i$ ensures that}
\fullv{such that by Lemma~\ref{lemma:dual}}
agents in $\agents_1$
observe each information set with the same probability, whether $i$ follows $\sigma_i''$ and $\sigma_i^*$ at $I_i$,
and agents in $\agents_2$ observe information sets consistent with agents following $\vsigma^*$ from the beginning.
This implies that, when $i$ follows $\sigma_i''$ at $I_i$ and the evolving graph is $G'$,
the expected utility of $i$ obtained in interactions with agents in $\agents_1$ is the same as when $i$ follows $\sigma_i^*$ (identical to $\sigma_i'$ at $I_i$).
Regarding interactions with agents $l \in \agents_2$, both $i$ and $l$ cooperate if $i$ follows $\sigma_i''$.
Now, suppose that for some $h \in I_i$ with $\mu^*(h \mid G',I_i) > 0$,
and run $r$ with $\pra{r(m'')}{\vsigma^*}{G',h} > 0$ for every $m'' > m$,
an agent $l \in \agents_2$ does not cooperate with $i$ in $r$.
Then, the expected utility of $i$ following $\sigma_i^*$ obtained in interactions with agents in $\agents_2$
is strictly lower than if $i$ always cooperated with those agents by following $\sigma_i''$, so we have
$$u_i(\vsigma^* \mid G',I_i) < u_i((\sigma_i'',\vsigma^*_{-i}) \mid G',I_i').$$
Thus, $\vsigma^*$ does not enforce accountability.

Now, consider that when $i$ follows $\sigma_i^*$ at $I_i$ agents in $\agents_2$ never punish $i$.
Fix round-$m$ information set $I_i'$ consistent with $\vsigma^*$ and $G'$.
Suppose that $\vsigma^*$ enforces accountability.
Then, at $I_i'$, $i$ must cooperate with $j$ according to $\sigma_i^*$.
Given that $i$ gains $1$ by defecting $j$,
there must be a utility loss in future interactions if $i$ 
deviates from $\sigma_i^*$ by defecting.
However, $i$ may follow a dual evasive strategy $\sigma_i'''$ at $I_i'$, 
where $i$ is never punished.
More precisely, $i$ can behave towards agents in $\agents_2$ as if the information
set were $I_i$ and $i$ followed $\sigma_i^*$ such that
by assumption agents in $\agents_2$ always cooperate with $i$.
In $\sigma_i'''$, $i$ can also behave towards agents in $\agents_1$
as if $i$ followed $\sigma_i^*$ at $I_i'$, such that these agents also cooperate with $i$.
Hence, $i$ obtains a strictly higher utility by following
$\sigma_i'''$ instead of $\sigma_i^*$ at $I_i'$.
Again, $\vsigma^*$ cannot enforce accountability.
This concludes the proof.
\end{proof}

\subsubsection{Unsafe Defections}

Fix a bounded protocol $\vsigma^*$ that enforces accountability,
and let $\rho$ be the maximum convergence time from an arbitrary state to a cooperation state.
Suppose that, in order to avoid the problem of ambiguous defections,
$\vsigma^*$ never requires an agent $i$ to defect $j$ if $j$ did not deviate from $\vsigma^*$.
Consider a scenario with three agents numbered 1 to 3,
where, agent 1 interacts with 2 and 3 in round~$1$,
agent 2 interacts with 3 in round~$\rho-1$,
and agent 3 interacts with 1 in round~$\rho$.
Suppose that (1) if 1 defects 2 or 3 in round~$1$, then the only timely opportunity
to punish 1 is in round~$\rho$,
and (2) if 2 defects 3 in round~$\rho-1$, then monitoring information revealing
this defection to agents that may punish 2 must be forwarded by 3 to 1
(e.g., due to all temporal paths crossing this edge).
Consider that $\vsigma^*$ requires 3 and 1 to mutually defect 
in round $3$ as a punishment for a defection of 1 in round~$1$.
Notice that $\vsigma^*$ still requires 2 to not defect 3.
After the defection of 1, if agent 2 knows that 3 will defect 1, regardless of whether
2 defects 3, then 2 will always gain from defecting 3, since any monitoring
information revealing this defection to other agents is lost when 3 defects 1.
Instead, suppose that 3 does not defect 1 after a defection of 2.
Now, we may have another problem: if 3 defects 1, its future neighbours
may never learn from the defection of 2, thus believing
that agents followed $\vsigma^*$ after round~$1$.
Since the system converges to a cooperation state after $\rho$,
future neighbours of 3 may never punish 3 for defecting 1,
in which case 3 gains from defecting 1, and $\vsigma^*$ cannot enforce accountability.

Formally, given $\rho$, agents $i$, $j$, $j'$, and $l$, rounds $m$ and $m'$, and $G\in \G^*$,
a \emph{temporal path} in $G$ from $i$-edge $(j,m)$
to $j'$-edge $(l,m')$ is a sequence of $c$ tuples $(j_k,j_{k+1},m_k)$ such that
(1) for every $k \le c$, $j_k \ne i$, $(j,m) \leadsto^G (j_k,m_k)$,
and $(j_k,j_{k+1})$ is an edge in $G^{m_k}$,
and (2) $j_{c+1} = l$ and $m_{c} < m'$. We say that $G$ is unsafe
if there exist rounds $m,m^1,m^2$, agents $i,j,l$,
$i$-edges $(j,m)$, $(l,m^2)$, and $j$-edge $(l,m^1)$, such that
(1) $l$ has no interactions after round $m^2$ and prior to round $m+\rho$,
and (2) $(l,i,m^2)$ is in all temporal paths from every $j$-edge $(l,m^1)$ to 
all $j$-edges after round $m^1$ and $i$-edges after round $m^2$.
A protocol $\vsigma^*$ is said to avoid ambiguous punishments
iff for every set $S$ of agents, protocol $\vsigma_{S}$ followed by agents in $S$,
$G \in \G^*$, agents $i\notin S$ and $j$, information set $I_i \in \I_i(G)$ 
consistent with $(\vsigma_S,\vsigma^*_{-S})$ and $G$,
and $h \in I_i$, $\vsigma^*$ requires $j$ to not defect $i$ at $h$.
Fix an unsafe $G \in \G^*$ and corresponding agents $i,j,l$ and rounds $m,m^1,m^2$.
Given a protocol $\vsigma^*$ that avoids ambiguous punishments, let $\sigma_i'$ be a strategy
where $i$ deviates from $\sigma_i^*$ exactly by defecting all round-$m$ neighbours.
We say that $\vsigma^*$ has an unsafe defection iff in every run $r$ 
of $(\sigma_i',\vsigma^*_{-i})$ in $G$, agent $l$ defects $i$ at round $m^2$.

Theorem~\ref{theorem:unsafe} shows that no bounded protocol $\vsigma^*$ 
that enforces accountability can have an unsafe defection if beliefs
are reasonable in the following sense. $\mu^*$ provides the simplest 
possible explanation to deviations,
by considering the smallest set of agents that deviated and the simplest
deviating strategy that explains the observations.
Specifically, fix $G \in \G^*$, agent $i$, and round-$m$ information set $I_i$.
We define $\mu^*(\cdot \mid G,I_i)$ as follows.
Agent $i$ deterministically selects the smallest set $S$
of agents for which there is a completely mixed protocol $\vsigma_S'$
such that $I_i$ is consistent with $(\vsigma_S',\vsigma_{-S}^*)$ and $G$.
Given $S$, $i$ deterministically selects a protocol $\vsigma_S'$
such that, for each agent $l \in S$, $l$ deviates in the least number of rounds,
and, if $l$ deviates in round $m$, then $l$ follows the same individual action
towards all round-$m$ neighbours.
It is easy to see that $\mu^*$ is consistent with $\vsigma^*$ and $\G^*$.
We just need to establish a total order regarding the protocols $(\vsigma_S',\vsigma_{-S}^*)$,
from the protocol with largest $S$ and largest number of deviations per agent,
to $\vsigma^*$, using some criteria to order protocols with the same size of $S$
and same number of deviations per agent such as the order of the identities of agents
and round numbers when agents deviate.
Every such protocol is completely mixed, and the sequence converges to $\vsigma^*$.

The proof of Lemma~\ref{theorem:unsafe} uses the notion of lenient evasive
strategy, defined as follows.
In a lenient evasive strategy, agent $i$ 
ignores a defection of some other agent $j$ towards
$i$ at round $m$, while still ensuring that no future neighbour of $i$
notices that $i$ ignored the defection of $j$.
Such strategy allows $i$ to take advantage of 
an unsafe defection in an interaction with $l$\shortv{ (see Appendix~\ref{app:unsafe})};
if the strategy $\sigma_i^*$ requires $i$ to defect $l$
iff $j$ does not defect $i$, then $i$ can gain by deviating
from $\sigma_i^*$ by defecting $l$ such that no agent will punish $i$ in the future.

Formally, consider an unsafe evolving graph $G \in \G^*$, 
\shortv{as defined in Appendix~\ref{app:unsafe},}
and fix corresponding agents $i,j,l$ and rounds $m,m^1,m^2$.
Fix also a corresponding $\vsigma^*$ that has unsafe omissions,
fix $\vsigma'$ where exactly $i$ deviates from $\vsigma^*$ by defecting all round-$m$ neighbours,
and fix round-$m^2$ information set $I_l \in \I_l(G)$ consistent with 
protocol $\vsigma''$ and $G$, where in $\vsigma''$
exactly $i$ and $j$ deviate from $\vsigma^*$ by defecting all
round-$m$ and round-$m^1$ neighbours, respectively.
We define a strategy $\sigma_l'''$ similar to single evasive strategies,
where $l$ pretends that $j$ did not defect at round $m$.
Specifically, let $S_{m'}$ be the set of agents $o$ not causally influenced
in $G$ by round-$m^1$ neighbours of $j$ between $m^1$ and $m'$,
and let $O_{m'}$ be the set of remaining agents $o$
for which there is an $l$-edge $(o',m'')$ with $m'' > m^2$
and $(o,m') \leadsto^G (o',m'')$.
Let $L_{m'} = S_{m'} \cup O_{m'} \cup \{i\}$.
Let $I_l'$ be an information set identical to $I_l$
in terms of the actions followed agents in $L_{m'} \setminus \{i\}$,
and consistent with $\vsigma'$ and $G$.
In $\sigma_l'$, at round $m^2$, agent $l$ selects $I_l'$
and follows actions with probability $\sigma_l^*(\cdot \mid I_l')$, defecting $i$.
For every $m' > m^2$, fix round-$m'$ information set $I_l$ and
corresponding $I_l'$ previously selected by $l$.
At $I_l$, agent $l$ follows actions with probability distribution $\sigma_l^*(\cdot \mid I_l')$.
Given round-$m'+1$ information set $I_l''$, $l$ selects $I_l'''$
compatible with $I_l'$ and with the round-$m'$ observations in $I_l''$.
We prove Lemma~\ref{lemma:lenient},
which shows that, for all rounds $m' > m^2$,
both $l$ and neighbours of $l$ follow actions consistent with $\vsigma'$,
using similar arguments to the proof of Lemma~\ref{lemma:hide}.

Let $\vsigma''' = (\sigma_l''',\vsigma^*_{-l})$. 
Notice that runs of $\vsigma'''$ also specify the selections of information sets by $l$.
We still use the same notation for runs. Given a run $r$,
an information set $I_{L_{m'}}$ compatible with $r$
specifies the observations of agents in $L_{m'} \setminus \{l\}$ and the selections of $l$.
We show in Lemma~\ref{lemma:lenient} that for all runs $r$ of $\vsigma'''$ and
$m' > m^2$, and round-$m'$ information set $I_{L_{m'}} \in \I_{L_{m'}}(G)$
compatible with $r$, $I_{L_{m'}}$ is consistent with $\vsigma'$.
By the definition of unsafe $G$, every round-$m'$ neighbour of $i$ 
is in $L_{m'}$, so every neighbour of $i$ believes that agents followed $\vsigma'$ after $m$.

\begin{lemma}
\label{lemma:lenient}
Suppose that beliefs are reasonable.
For every round $m' \ge m^2$, history $h \in I_l$ with $\mu^*(h \mid G,I_l) > 0$,
and run $r$ of $\vsigma''$ after $h$ in $G$, the round-$m'$ information set $I_{L_{m'}}$
compatible with $r$ is consistent with $\vsigma'$.
\end{lemma}
\begin{proof}
The proof is by induction on $m'$. Consider $m' = m^2$.
Fix $h \in I_l$ with $\mu^*(h \mid G,I_l) > 0$.
Since beliefs are reasonable and $h$ is consistent with $\vsigma''$ and $G$,
$\mu^*(h \mid G,I_i)$ is given by conditioning on agents following $\vsigma''$.
Using the same arguments as in the proof of Lemma~\ref{lemma:hide},
we can show that for every agent $o \in S_{m'}$, the information set $I_o$
corresponding to $h$ is consistent with $\vsigma'$.
Moreover, since every temporal path from round-$m^1$ $j$-edges to $l$-edges
after $m^2$ crosses $(l,i,m^2)$, we have $O_{m'} = \emptyset$.
Finally, $l$ selects $I_l'$, which is consistent with $\vsigma'$ by construction.
This proves the base case.

Now, suppose that the hypothesis holds for $m' \ge m^2$
and fix round-$m'$ history $h'$ such that there is $h \in I_l$
with $\mu^*(h \mid G,I_l) > 0$ and $\pra{h'}{\vsigma'}{G,h} > 0$.
Let $I_l'$ be the round-$m'$ information set selected by $l$.
Fix $\va'$ with $\pra{\va'}{\vsigma'}{G,h'} > 0$.
Every agent $o \in L_{m'}$ follows $a_o$ with probability $\sigma_o^*(a_l \mid I_o)$ if $o \ne i$,
where $h \in I_o$, or $o$ follows $a_o$ with probability $\sigma_o^*(a_o \mid I_l')$ if $o = l$.
Since $I_o$ and $I_l'$ are consistent with $\vsigma'$ and $G$,
$\va'$ is also consistent with $\vsigma'$ (which is equivalent to $\vsigma^*$ after $m$) and $G$.
By definition of unsafe $G$, for every $o \in L_{m'+1} \setminus \{i\}$,
the round-$m'$ neighbours of $o$ are in $L_{m'}$, so the round-$m'+1$ information set $I_o'$
corresponding to $(h,\va)$ is consistent with $\vsigma'$ and $G$.
In addition, $l$ selects $I_l''$ compatible with $I_l'$ and the actions of the round-$m'$
neighbours of $l$ in $\va$. Since every round-$m'$ neighbour of $l$ is also in $L_{m'}$,
$I_l''$ is consistent with $\vsigma'$ and $G$. This concludes the proof.
\end{proof}

We can now prove Lemma~\ref{theorem:unsafe}.

\begin{restatable}{theorem}{theoremunsafe}
\label{theorem:unsafe}
If $\vsigma^*$ is bounded and has unsafe defections and beliefs are reasonable,
then $\vsigma^*$ does not enforce accountability.
\end{restatable}
\begin{proof}
Suppose that beliefs are reasonable and that $\vsigma^*$ has unsafe omissions,
and fix corresponding $G \in \G^*$, $i$, $j$, $l$, $m$, $m^1$, and $m^2$.
Fix a reasonable belief system $\mu^*$.
We show that $\vsigma^*$ is not a \oape{\G^*} with $\mu^*$.

Let $\sigma_i'$ be a strategy that differs from $\sigma_i^*$ exactly in that
$i$ defects all round-$m$ neighbours, and let $\vsigma' = (\sigma_i',\vsigma^*_{-i})$.
Since $\vsigma^*$ avoids ambiguous punishments, $i$ must not defect $j$ and $l$ at 
round-$m$ information sets consistent with $\vsigma^*$ and $G$ (and with $\vsigma'$).
This implies that, in every run $r$ of $\vsigma'$ in $G$, prior to round $m^1$,
both $j$ and $l$ know that $i$ deviated from $\sigma_i^*$.
That is, for all $o \in \{j,l\}$, round-$m^1$ information set $I_o \in \I_o(G)$
consistent with $\vsigma'$ and $G$, and history $h \in I_o$,
agent $i$ deviates from $\sigma_i^*$ in $h$ at round $m$, by defecting $o$.
This implies that in the definition of $\mu^*$, the set $S$ must include $i$,
and $i$ must deviate by defecting all round-$m$ neighbour.s
Since $I_o$ is consistent with $\vsigma'$ and $G$, 
the simplest explanation for $I_o$
is that exactly $i$ deviated by defecting all round-$m$ neighbours, so
$\mu^*( h \mid G,I_o)$ is defined by conditioning on $G$ and agents following $\vsigma'$.

Let $\sigma_j''$ be an evasive strategy where $j$ deviates from $\sigma_j^*$
by defecting $l$ at $m^1$, and later $j$ behaves as if $j$ did not defect $l$.
Let $\vsigma'' = (\sigma_i',\sigma_j'',\vsigma^*_{-i,j})$.
First, suppose that in every run of $\vsigma''$ agent $l$ still defects $i$ at round $m^2$.
Fix round-$m^1$ $I_j$ consistent with $\vsigma'$ and $G$.
We can show that if agents follow $\vsigma''$,
then for every round $m' > m^1$, $j$-edge $(o,m')$,
and round-$m'$ information set $I_o \in \I_o(G)$, we have
\begin{equation}
\label{eq:unsafe}
\sum_{h \in I_j} \mu^*(h \mid G,I_j) \sum_{r : r_o(m') = I_o} \pra{r}{\vsigma''}{G,h} = 
\sum_{h \in I_j} \mu^*(h \mid G,I_j) \sum_{r : r_o(m') = I_o} \pra{r}{\vsigma'}{G,h},
\end{equation}
where $r_o(m')$ is the round-$m'$ information set with $r(m') \in r_o(m')$.
To see this, given round $m' > m^1$, let $S_{m'}$ be the set of agents $o$ such that
$(l,m^1) \leadsto^G (o,m')$ is false.
We can show using the same
arguments as Lemma~\ref{lemma:hide}\shortv{ (see Appendix~\ref{app:evasive})} that
(\ref{eq:unsafe}) holds for every $m' > m$ and $o \in S_{m'}$.
Since $G$ is unsafe, every neighbour of $j$ up to round and including $m^2$ is in $S_{m'}$.
This proves (\ref{eq:unsafe}) for $m' \le m^2$.

Now, let $O_{m'}$ be the set of agents $o$ such that, for some $o'$ and $m'' < m'$,
$(o',o,m'')$ is in a temporal path from round-$m^1$ $j$-edges to $j$-edges after $m^1$.
We can show using induction that (\ref{eq:unsafe}) also holds
for every $m' > m^2$ and $o \in O_{m'}$.
Since $G$ is unsafe, for every round $m' > m^2$, 
$l$ cannot be in $O_{m'}$ and every round-$m'$ neighbour of $j$ is in $S_{m'} \cup O_{m'}$.
This implies (\ref{eq:unsafe}) for all $j$-edge $(o,m')$ with $m' > m^1$.

We now prove the induction.
First, consider $m' = m^2 + 1$. By definition of unsafe $G$,
we have $O_{m'} = \{i\}$ and $i \in S_{m'}$.
As we have seen before, in runs of $\vsigma'$ and $\vsigma''$ in $G$,
$i$ observes each round-$m^2$ information with the same probability,
and, since $i$ follows $\sigma_i^*$ after round $m$,
$i$ also follows each round-$m^2$ action with the same probability.
The same applies to all round-$m^2$ neighbours of $i$, except $l$.
However, since $l$ always defects $i$, it holds that $i$ observes
each round-$m^2+1$ information set with the same probability,
whether agents follow $\vsigma'$ or $\vsigma''$.
This proves (\ref{eq:unsafe}) for $m' = m^2 + 1$.
Continuing inductively, for every $o \in O_{m'+1}$, by
the definition of unsafe $G$, both $o$ and
every round-$m'$ neighbour of $o$ are in $S_{m'} \cup O_{m'}$, so the
result follows from the same arguments as above and the hypothesis.

Now, suppose that there is a round-$m^1$ information set $I_l \in \I_l(G)$
consistent with $\vsigma''$ such that $l$ does not defect $i$.
In every run of $\vsigma'$ in $G$ agent $j$ must not defect $l$ at $m^1$.
This implies that at $I_l$ agent $l$ knows that both $i$ and $j$ deviated from $\sigma_j^*$.
Again, $\mu^*(\cdot \mid G,I_l)$ is defined by conditioning on $\vsigma''$.
\fullv{As shown in Lemma~\ref{lemma:lenient},}
\shortv{As shown in Appendix~\ref{app:variant},}
$l$ may follow a lenient evasive strategy $\sigma_l'''$
such that in every round $m' > m^2$ the neighbours of $l$ behave as if if all agents followed $\vsigma'$
after $m$. Since $\vsigma^*$ and $\vsigma'$ are equivalent after $m$,
and since the system stabilizes to a cooperation state by round $m+\rho$
and $l$ has no interactions after $m^2$ and prior to $m+\rho$,
in every interaction after $m^2$ between $l$ and a neighbour $o$, both 
$l$ and $o$ cooperate. Thus, the expected utility in all future interactions is $\beta - 1-\alpha$
when $l$ follows $\sigma_l'''$. If $l$ follows $\sigma_l^*$
at $I_l$, then the expected utility in all of those interactions is at most $\beta - 1 - \alpha$.
Since $l$ gains at least $1$ by defecting $i$ at $I_l$,
$l$ gains by deviating from $\sigma_l^*$ at $I_l$,
hence $\vsigma^*$ cannot enforce accountability.
This concludes the proof.
\end{proof}

}

\subsection{Need for Eventual Distinguishability}
\label{sec:dist}

We identify a necessary restriction on $\G^*$ to enforce accountability with safe-bounded protocols
in general pairwise exchanges, called eventual distinguishability.
\fullv{Later, we generalize the results to safe (possibly non-bounded) protocols.}
\shortv{Due to lack of space, we only provide an informal definition of this condition
and an intuition for its necessity, deferring the formal definition and proof to Appendix~\ref{app:dist};
we defer the generalized proof for non-bounded protocols to the full paper.}
Roughly speaking, we say that the \emph{adversary is restricted by eventual distinguishability}
if for every evolving graph $G$ in $\G^*$ and agent $i$,
eventually $i$ stops having interactions in $G$ where $i$ defects some neighbour $j$,
and then expects to be punished by two or more agents for that defection,
because each of those agents cannot distinguish between two evolving graphs
where it should and should not punish $i$.
If the adversary is not restricted by eventual distinguishability
and agent $i$ keeps defecting its neighbours, then the number
of punishments that $i$ expects to receive in future rounds grows unboundedly with time.
Consequently, the protocol cannot be bounded,
since there is no bound on the time it takes for punishments to end after agents stop deviating.

Specifically, the formal definition of eventual distinguishability has two parts:
(1) a definition of \emph{indistinguishable rounds}, in which a defection of some
agent is matched by more than one future punishment,
and (2) the requirement that indistinguishable rounds must eventually stop occurring.
Fig.~\ref{fig:indist} depicts an example of an indistinguishable round.
There are three agents numbered $1$ to $3$.
The adversary may generate three alternative evolving graphs $G^1$, $G^2$, and $G^3$.
In round $1$, agent $1$ interacts with agent 2.
In round $2$, agent 2 interacts with agents 3.
In round $3$, the interactions depend on the evolving graph:
($G^1$) $i$ interacts only with $2$, ($G^2$) $i$ interacts only with $3$,
and ($G^3$) $i$ interacts with both $2$ and $3$.
In every safe protocol that enforces accountability in general pairwise exchanges,
1 must send a message to 2 in round 1, 
so 1 gains by defecting 2.
Suppose that round 3 is the only opportunity to 
punish $i$ for defecting 2 in round $1$ in a timely fashion.
In an equilibrium protocol, the immediate gain $1$ of defecting must be lower than the
future loss $\beta$ of being punished (or else 1 would gain by deviating).
Since $\beta$ can be arbitrarily close to $1$\,\footnote{Recall that we only assume that $\beta > 1+\alpha$ and $\alpha$ can be small.},
1 must be (deterministically) punished by 2 in $G^2$ and must
be punished by 3 in $G^3$ (in which case, 2 must inform 3 of the defection of 1 in round $2$).
Suppose that agents only know the identities of their neighbours.
Then, neither agent 2 can distinguish $G^3$ from $G^1$ nor can
agent 3 distinguish $G^3$ from $G^2$ at round~$3$.
Thus, they both have to punish 1 in $G^3$, and
the expected number of punishments
must increase by $2$, even though 1 only defects one neighbour.
This shows that there is a round-2 information set $I_1 \in \I_1(G^3)$
such that, according to the information in $I_1$ available to $1$,
the expected number of punishments of $1$ conditional on $G^3$ and $I_1$ is $2$.
If this type of interactions keeps occurring in $G^3$,
then, for an arbitrarily large number $c$, there is an information set $I_1 \in \I_1(G^3)$
such that, basing on the information in $I_1$, agent
1 expects to be punished by at least $c$ neighbours after $I_1$.
The fact that only agent 1 deviates plays a key role here:
by the definition of equilibrium, we only need to ensure that
1 never gains by defecting a neighbour at round 1 provided that other agents do not
deviate afterwards. If multiple agents
deviate after round $1$ (e.g., agent 2 defects 3 in round 2),
then 1's defection in round $1$ can be forgiven,
since such behaviour is unexpected at the time 1 decides to defect 2 in round 1.
The assumption that protocols are safe is also crucial,
since non-safe protocols are not susceptible to the aforementioned problem.
Unfortunately,\shortv{ as we show in the full paper,}
we cannot in general devise a non-safe protocol
that enforces accountability.

\begin{figure}
\begin{center}
 \includegraphics[scale=0.2]{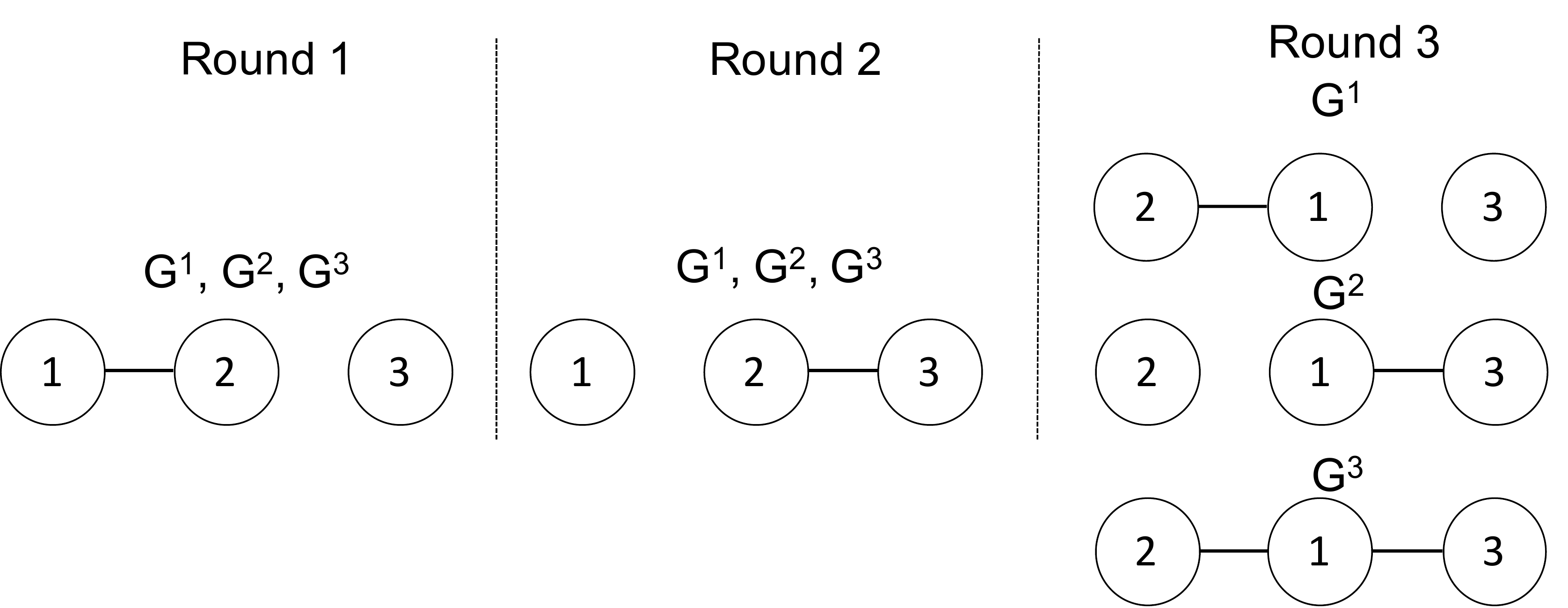}
\end{center}
\vspace{-0.2in}
\caption{$G^3$ is indistinguishable from $G^1$ ($G^2$) to 2 (3) at round $3$.}
\label{fig:indist}
\end{figure}

\shortv{In Appendix~\ref{app:dist}, we prove Theorem~\ref{theorem:dist}.
The proof shows that, if the adversary is not restricted by eventual distinguishability,
then, for some agent $i$ and evolving graph $G$ in $\G^*$, if $i$ always defects all neighbours in $G$, then, 
in an indistinguishable round, the expected number of future punishments after the defections
is strictly larger than before, whereas in all remaining rounds the expected number of punishments never decreases.
This is sufficient to show that the expected
number of punishments is unbounded, and hence that no safe protocol that enforces
accountability in general pairwise exchanges can be bounded.

}\fullv{
We now provide a rigorous definition of an adversary restricted by eventual distinguishability.
Given an evolving graph $G$, agent $i$, and constant $\rho$,
let $\po{G}{\rho}{m}$ be the set of PO's $(l,m')$ of $i$ in $G$ for $i$-edges $(j,m'')$ such that $m'' \ge m$. 
We say that round $m$ is $(G,i,\rho)$-indistinguishable iff
there are two evolving graphs $G^1,G^2 \in \G^*$ such that
(1) for every $G' \in \{G^1,G^2\}$ and $(j,m') \in \po{G'}{\rho}{m}$, $G'$ is indistinguishable from $G$ to $j$ at $m'$,
(2) $|\po{G^1}{\rho}{m} \cap \po{G^2}{\rho}{m}| < k$ where $k$ is the number of $i$-edges in round $m$,
and (3) $\po{G^1}{\rho}{m} \cup \po{G^2}{\rho}{m} = \po{G}{\rho}{m}$.
We say that the adversary is restricted by \emph{eventual distinguishability} iff
there is a constant $\rho>0$ and a round $m^*$ such that for all $G \in \G^*$ and agent $i$,
every round $m > m^*$ is \emph{not} $(G,i,\rho)$-indistinguishable.

This definition has two parts: (1) the definition of $(G,i,\rho)$-indistinguishable round
and (2) the requirement that, eventually, no round is $(G,i,\rho)$-indistinguishable.
We have already explained that (2) is necessary to keep punishments bounded.
Part (1) is a generalization of the scenario depicted in Fig.~\ref{fig:indist}
for the worst-case scenario where some agent always defects its neighbours.
To understand its definition, recall that the problem identified in the example
was that there were two different agents 2 and 3 who were responsible for punishing 1 in a timely fashion,
after 1 defected 2 in round 1, in evolving graphs $G^1$ and $G^2$,
and who could not distinguish a third evolving graph $G^3$ from $G^1$ and $G^2$, respectively.
More generally, this happens when, for some agent $i$, round $m$, and bound $\rho$ on the delay of punishments,
(1) there are sets $S_1$ and $S_2$ of agents responsible for punishing $i$ in $G^1$ and $G^2$, respectively,
in the $\rho$ rounds after $i$ defects some (or all) neighbours in round $m$ (in our example, $S_1 = \{2\}$ and $S_2 = \{3\}$),
(2) $S_1$ and $S_2$ do not intersect, and (3) for $l =1,2$, no agent in $S_l$ can distinguish $G^3$ from $G^l$.
We further generalize this intuition by showing that the problem arises even if $S_1$ and $S_2$ intersect;
they just cannot intersect in more than $k$ agents, where $k$ is the degree of $i$ in round $m$.
We also show that, if $S_1$ and $S_2$ contain all the agents that can punish $i$ 
no later than $\rho$ rounds after $m$, then no protocol can avoid the aforementioned problem.
Given that we are considering the worst-case scenario where $i$ always defects its neighbours,
this is exactly the case when $S_1$ and $S_2$ are the agents in 
$\po{G^1}{\rho}{m}$ and $S_2 = \po{G^2}{\rho}{m}$, respectively.

To see this, notice that when $i$ defects a round-$m$ neighbour $j$,
if $(l,m')$ is a PO of $i$ for $(j,m)$, then $l$ can clearly learn of the deviation of $i$ 
and punish $i$ in round $m'$, so $l$ is an agent capable of punishing $i$ for defecting a round-$m$ neighbour.
However, since $i$ may also defect other agents in rounds after $m$,
these are not the only agents capable of punishing $i$ for a defection in round $m$.
Consider the following scenario where the neighbours of $i$ in $m$ and $m+1$ are
$j$ and $j'$ respectively. Consider three PO's of $i$: (1) $(l^1,m^1)$ is a PO for both $(j,m)$
and $(j',m+1)$, (2) $(l^2,m^2)$ is a PO for $(j',m+1)$ but not $(j,m)$,
and, similarly, (3) $(l^3,m^3)$ is a PO for $(j',m+1)$ but not $(j,m)$.
In this case, we can devise a punishment scheme were $i$ is punished
at least once for each defection, and at most twice for defecting both $j$ and $j'$:
(i) if $i$ only defects $j$, then only $l^1$ punishes $i$; and
(ii) if $i$ defects $j'$, then $l^2$ and $l^3$ punish $i$,
whereas $l^1$ does not punish $i$, whether $i$ defects $j$.
(ii) is possible, because $l^1$ can be informed of the defection of $i$ towards $j'$.
In this case, we can say that $i$ is punished for defecting $j$ by $l^2$, even though 
$(l^2,m^2)$ is not a PO of $i$ for $(j,m)$; it is a PO for an $i$-edge from a round $m' > m$.
In other words, in general, $i$ may be punished for defecting a round-$m$
neighbour by an agent $l$, in round $m'$, if $(l,m')$ is a PO of $i$ for some $i$-edge
$(j',m')$ with $m' \ge m$. So, the set of $i$-edges corresponding to agents that may punish $i$
for defecting a round-$m$ neighbour in evolving graph $G$,
in the $\rho$ rounds after $m$, is exactly $\po{G}{\rho}{m}$.

Theorem~\ref{theorem:dist} proves the need for eventual distinguishability, in order
to devise safe-bounded protocols that enforce accountability in general pairwise exchanges.
Fix a safe protocol $\vsigma^*$ that enforces accountability in general pairwise exchanges,
fix a belief system $\mu^*$ consistent with $\vsigma^*$ and $\G^*$,
and let $\rho$ be the maximum time it takes for $\vsigma^*$ to converge to a cooperation state.
\shortv{Given an agent $i$, $G \in \G^*$, round-$m$ information set $I_i \in \I_i(G)$,
and $i$-edge $(j,m')$ with $m' > m$, $\vsigma$ defines a 
probability $\pra{j,m'}{\vsigma}{G,I_i}$ of $j$ punishing $i$ at round $m'$.
Let $\mathbb{E}^{\vsigma,\rho}[P_i \mid G]$ be the expected number of 
punishments of $i$ in the next $\rho-1$ rounds, defined as
$$\mathbb{E}^{\vsigma,\rho}[P_i \mid G,I_i] = \sum_{\mbox{$i$-edge }(j,m'): m < m' < m+\rho} \pra{j,m'}{\vsigma}{G,I_i}.$$}

The proof is divided into three parts. First, we show in Lemma~\ref{lemma:minpuns}
that when agent $i$ defects $k$ neighbours in a round the expected number of 
punishments in the next $\rho-1$ rounds must increase at least by $k$.
Then, Lemma~\ref{lemma:acum} shows that, if $i$ constantly defects all neighbours
and the adversary is not restricted by eventual distinguishability,
then the expected number of punishments of $i$ grows without bound.
Finally, in the proof of Theorem~\ref{theorem:dist}, we identify a contradiction
between Lemma~\ref{lemma:acum} and the requirement of self-stabilization in bounded time.

\begin{lemma}
\label{lemma:minpuns}
For every $G \in \G^*$, agent $i$,
information set $I_i \in \I_i(G)$, and protocol $\vsigma' = (\sigma_i',\vsigma^*_{-i})$
where $i$ omits messages to $k$ neighbours at $I_i$ and follows $\sigma_i^*$ afterwards, we have
$$\mathbb{E}^{\vsigma',\rho}[P_i \mid G,I_i] \ge \mathbb{E}^{\vsigma^*,\rho}[P_i \mid G,I_i] + k.$$
\end{lemma}
\begin{proof}
Suppose that there is $G$, $i$, $I_i$, and $\vsigma'$ where $i$ omits messages to $k$ neighbours at $I_i$ and
\begin{equation}
\label{eq:minpuns1}
\mathbb{E}^{\vsigma',\rho}[P_i \mid G,I_i] < \mathbb{E}^{\vsigma^*,\rho}[P_i \mid G,I_i] + k.
\end{equation}

At $I_i$, $i$ avoids the cost $k$ of sending messages by following $\sigma_i'$ instead of $\sigma_i^*$,
whereas the expected benefits and costs of receiving messages and being punished in round $m$ are the same,
whether $i$ follows $\sigma_i^*$ or $\sigma_i'$.
In every round $m' > m$ and interaction with round-$m'$ neighbour $j$,
both $i$ and $j$ either cooperate with or punish each other, 
given that $\vsigma^*$ is safe.
Thus, $i$ incurs the fixed cost $1+\alpha$ of sending and receiving messages,
and the expected utility loss of following $\sigma_i'$ instead
of $\sigma_i^*$ in an interaction with $j$
is upper bounded by the probability of $j$
punishing $i$ times the maximum loss of $\beta$.
Let $S$ be the largest set of $i$-edges such that
$\Delta(j,m') > 0$ for every $(j,m') \in S$, where
$\Delta(j,m') = \pra{j,m'}{\vsigma'}{G,I_i} - \pra{j,m'}{\vsigma^*}{G,I_i}$.
Since $\vsigma^*$ is self-stabilizing in at most $\rho$ rounds, for every $(j,m') \in S$,
we have $m' - m < \rho$, and
\begin{equation}
\label{eq:lemma2}
\begin{array}{lll}
&u_i(\vsigma' \mid G,I_i) - u_i(\vsigma^* \mid G,I_i) & \ge\\
\ge & k - \sum_{(j',m') \in S} \delta^{m' - m}\Delta(j',m') \beta & \ge\\
\ge & k - \sum_{(j',m') \in S} \Delta(j',m') \beta & \ge\\
\ge &  k - (\mathbb{E}^{\vsigma',\rho}[P_i \mid G,I_i] - \mathbb{E}^{\vsigma^*,\rho}[P_i \mid G,I_i]) \beta.
\end{array}
\end{equation}

By (\ref{eq:minpuns1}) and the assumption that $\vsigma^*$ is independent
of the utilities, there is a value of $\beta > 1$ sufficiently
close to $1$ such that, for some $G$ and $I_i$, the utility difference between
$i$ following $\sigma_i'$ and $\sigma_i^*$ is strictly positive,
regardless of the value of $\delta$,
so $i$ gains by deviating from $\sigma_i^*$ at $I_i$.
This is a contradiction to $\vsigma^*$ being a \oape{\G^*} for every $\beta > 1$,
thus proving the result.
\end{proof}

\fullv{The following lemma is useful in the proof.
Let $\sigma_i'$ be a strategy where $i$ defects round-$m$
neighbours but follows $\sigma_i^*$ afterwards.
We show that, if $(l,m')$ is not a PO of $i$ in $G$
for $i$-edges $(j,m'')$ with $m'' \ge m$,
then $(l,m')$ never learns from the round-$m$ defections of $i$.

\begin{lemma}
\label{lemma:needpo}
For every round $m' > m$ with $m' \le m$, set $S$ of agents $l \ne i$ such that
$(l,m')$ is not a PO of $i$ in $G$ for $i$-edges between $m$ and $m'$, and round-$m'$ information set $I_{S} = \cap_{l \in S} I_l$,
we have
\begin{equation}
\label{eq:needpo}
\sum_{h \in I_i} \mu^*(h \mid G,I_i) \sum_{h' \in I_{S}} \pra{h'}{\vsigma'}{G,h} =
\sum_{h \in I_i} \mu^*(h \mid G,I_i)\sum_{h' \in I_S} \pra{h'}{\vsigma^*}{G,h}.
\end{equation}
\end{lemma}
\begin{proof}
Fix $m'$ and $S$. We have that for all $l \in S$ and $i$-edge $(j,m)$,
$(j,m) \leadsto^G (l,m')$ is false and $i$ does not interact with $l$ prior to $m'$.
Given round $m''$, let $S_{m''}$ be the set of agents $l$ such that 
$(l,m'')$ is not a PO of $i$ in $G$ for $i$-edges between $m$ and $m''$.
At every round $m'' > m$ with $m'' < m'$, the information set
observed by agents in $S_{m''}$, depends only on the actions of 
agents not in $S_{m''-1}$. Using this fact, it is direct
to show the result by induction on $m''$.
\end{proof}
}

We can now prove Lemma~\ref{lemma:acum}.

\begin{lemma}
\label{lemma:acum}
If the adversary is not restricted by eventual distinguishability, then there exist
$G \in \G^*$ and agent $i$ such that for every $c > 0$, there is $I_i \in \I_i(G)$ such that
$$\mathbb{E}^{\vsigma^*,\rho}[P_i \mid G,I_i] \ge c.$$
\end{lemma}
\begin{proof}
Suppose that the adversary is not restricted by eventual distinguishability.
There is $G \in \G^*$, $i$, and an infinite sequence of rounds $m_1, m_2,\ldots$
such, that for every $k$, $m_k$ is not $(G,i,\rho)$-distinguishable.
Given round $m$, let $\vsigma' = (\sigma_i',\vsigma^*_{-i})$ be the 
protocol where $i$ always defects up to and including round $m$, and follows $\sigma_i^*$ afterwards.
By Lemma~\ref{lemma:minpuns}, for every $G' \in \G^*$ and round-$m$ information set $I_i \in \I_i(G)$,
\begin{equation}
\label{eq:indist2}
\mathbb{E}^{\vsigma',\rho}[P_i \mid G',I_i] \ge \mathbb{E}^{\vsigma^*,\rho}[P_i \mid G',I_i] + k,
\end{equation}
where $k$ is the number of round-$m$ neighbours of $i$.
\fullv{By Proposition~\ref{lemma:needpo},}
\shortv{We show in the full paper that}
only agents agents in $\po{G'}{\rho}{m}$ vary the probability of punishing $i$,
hence we have
$$\mathbb{E}^{\vsigma',\rho}[P_i \mid G',I_i] = \mathbb{E}^{\vsigma^*,\rho}[P_i \mid G',I_i]  + \sum_{(j,m') \in \po{G'}{\rho}{m}} \Delta(j,m').$$

In particular, this is true for $m=m_k$ and $G' \in \{G^1,G^2\}$ such that
(1) for every $(j,m) \in \po{G'}{\rho}{m}$, $G'$ is indistinguishable from $G$ to $j$ at $m$,
(2) $|\po{G^1}{\rho}{m} \cap \po{G^2}{\rho}{m}| < k$,
and (3) $\po{G^1}{\rho}{m} \cup \po{G^2}{\rho}{m} = \po{G}{\rho}{m}$.
\fullv{By Proposition~\ref{prop:indist},
every such $j$ must punish $i$ with the same probability both in $G$ and $G'$,
because $G'$ is indistinguishable from $G$ to $j$ at $m$,
so by (\ref{eq:indist2}) we can write}
\shortv{In the full paper, we show that every such $j$ must
punish $i$ with the same probability both in $G$ and $G'$,
because $G'$ is indistinguishable from $G$ to $j$ at $m$,
so by (\ref{eq:indist2}) we can write}
$$
\begin{array}{lll}
& \mathbb{E}^{\vsigma',\rho}[P_i \mid G,I_i] &= \\
=& \mathbb{E}^{\vsigma^*,\rho}[P_i \mid G,I_i] + &\\
 &+ \sum_{(j,m') \in \po{G^1}{\rho}{m}} \Delta(j,m)
  + \sum_{(j,m') \in \po{G^2}{\rho}{m}} \Delta(j,m) -&\\
  &- \sum_{(j,m') \in \po{G^1}{\rho}{m} \cap \po{G^2}{\rho}{m}} \Delta(j,m) & \ge\\
\ge & \mathbb{E}^{\vsigma^*,\rho}[P_i \mid G,I_i] + 2k -(k - 1) & \ge \\
\ge & \mathbb{E}^{\vsigma^*,\rho}[P_i \mid G,I_i] + k + 1.
\end{array}
$$

\fullv{By Proposition~\ref{prop:preserve},}
\shortv{By the properties of consistent beliefs shown in the full paper,}
there is round-$m+1$ $I_i' \in \I_i(G)$ such that
$$\mathbb{E}^{\vsigma',\rho}[P_i \mid G,I_i'] \ge \mathbb{E}^{\vsigma^*,\rho}[P_i \mid G,I_i] 
+ k +1 - k \ge \mathbb{E}^{\vsigma^*,\rho}[P_i \mid G,I_i] + 1.$$
This implies that, after every round $m_k$, there is an information set such that
the expected number of punishments of $i$
increases by at least $1$ if $i$ follows $\sigma_i'$,
while it does not decrease after every other round.
It is easy to show using induction that, for every round $m$,
there is a round-$m$ information set $I_i \in I_i(G)$
such that the expected number of punishments of $i$ at $I_i$ is at least $k^*$,
where $k^*$ is the largest $k$ with $m_k < m$.
We can have $k^*$ to be arbitrarily large, by arbitrarily increasing $m$,
hence the result follows.
\end{proof}

Finally, we prove Theorem~\ref{theorem:dist}.
}

\begin{restatable}{theorem}{theoremdist}
\label{theorem:dist}
If the adversary is not restricted by eventual distinguishability,
then there is no safe-bounded protocol that enforces accountability in general pairwise exchanges.
\end{restatable}
\fullv{
\begin{proof}
The proof is by contradiction. Suppose that the adversary is not restricted by eventual distinguishability.
By Lemma~\ref{lemma:acum},
there is $G$ and $i$ such that
for every $c > 0$, there is a round $m$ and round-$m$ $I_i \in \I_i(G)$ such that
$\mathbb{E}^{\vsigma^*,\rho}[P_i \mid G,I_i] \ge c$.
This is true for $m$ large enough such that
$c$ is larger than the number of $i$-edges between $m$ and $m+\rho$.
This is a contradiction to $\vsigma^*$ converging to a cooperation state in $\rho$ rounds.
\end{proof}
}

\fullv{
\subsection{Need for Distinguishability in Non-bounded Protocols}
We now prove a slightly weaker result than the previous section,
for the more general case of safe but possibly non-bounded protocols.
We show that, in order to enforce accountability in general pairwise exchange
protocols, the adversary must be restricted by \emph{frequent distinguishability},
which we now define. As in the definition of eventual distinguishability,
frequent distinguishability has two parts: (1) the definition of indistinguishable
round, and (2) the requirement that indistinguishable rounds do not occur frequently,
where by frequently we mean they keep occurring at every interval of $c$ rounds
for some constant $c$. Part (1) is the same as in eventual distinguishability.
Part (2) differs exactly in that indistinguishable rounds may occur infinitely often,
but interval between two consecutive indistinguishable rounds cannot be bounded.
Formally, we say that the \emph{adversary is restricted by frequent distinguishability}
iff for every constant $c > 0$, there is $\rho> 0$ such that, for all $G \in \G^*$, and agent $i$, there is 
a sequence of $c$ consecutive rounds in $G$ that are not $(G,i,\rho)$-indistinguishable.

Theorem~\ref{theorem:dist-gen} proves the need for frequent distinguishability.
The proof is almost identical to that of Theorem~\ref{theorem:dist},
in that we show by contradiction that, if the adversary is not restricted
by frequent distinguishability, then the expected number of punishments of some
agent is unbounded. The only difference is that, in every indistinguishable round,
the expected number of punishments grows by $1-\epsilon$ for an arbitrarily small $\epsilon$.

We first show that, for an arbitrarily small $\epsilon > 0$,
there is a sufficiently large delay $\rho$ such that,
if $i$ deviates from a safe equilibria protocol $\vsigma^*$ at round $m$
by first defecting $k$ neighbours, the expected
number of punishments in the next $\rho$ rounds must be higher 
by at least $k - \epsilon$ than following $\vsigma^*$.

\begin{lemma}
\label{lemma:minpuns-gen}
For every $\epsilon > 0$, there is $\rho > 0$ such that, for all $G \in \G^*$, agent $i$,
information set $I_i \in \I_i(G)$, and protocol $\vsigma' = (\sigma_i',\vsigma^*_{-i})$
where $i$ omits messages to $k$ neighbours at $I_i$ and follows $\sigma_i^*$ afterwards, we have
$$\mathbb{E}^{\vsigma',\rho}[P_i \mid G,I_i] \ge \mathbb{E}^{\vsigma^*,\rho}[P_i \mid G,I_i] + k - \epsilon.$$
\end{lemma}
\begin{proof}
By contradiction, if the Lemma is false, then there exists $\epsilon > 0$
such that, for an arbitrarily large $\rho$, we have
$$\mathbb{E}^{\vsigma',\rho}[P_i \mid G,I_i] < \mathbb{E}^{\vsigma^*,\rho}[P_i \mid G,I_i] + k - \epsilon.$$
By following $\sigma_i'$, $i$ gains $k$ while losing at most $\delta(k-\epsilon) \beta$
in the next $\rho$ rounds. In the rounds after $m+\rho$, $i$ loses at most $\beta$ per round.
Therefore, the utility difference is at most
$$k - \delta \beta ( k - \epsilon + \frac{\delta^{\rho}}{1-\delta}).$$
The constant $\rho$ can be arbitrarily large
so that $\epsilon > \delta^{\rho}/(1-\delta)$, since $\epsilon$ and $\delta$ are constant on $\rho$.
If this is the case, then $i$ gains by following $\sigma_i'$, and $\vsigma^*$ cannot
be a \oape{\G^*}. This is a contradiction, concluding the proof.
\end{proof}

We now prove that, if the adversary is not restricted by frequent distinguishability,
then the expected number of punishments of some agent is unbounded.

\begin{lemma}
\label{lemma:acum-gen}
If the adversary is not restricted by frequent distinguishability, then, for every $\rho > 0$,
there is $G \in \G^*$ and agent $i$ such that for every $b > 0$, there is $I_i \in \I_i(G)$ such that
$$\mathbb{E}^{\vsigma^*,\rho}[P_i \mid G,I_i] \ge b.$$
\end{lemma}
\begin{proof}
Suppose that the adversary is not restricted by frequent distinguishability.
There exists $c > 0$ such that, for every $\rho$, there is $i$ and $G$ such that frequently,
i.e. at most every $c$ rounds, some round is $(G,i,\rho)$-indistinguishable.
Similar to the proof of Lemma~\ref{lemma:minpuns},
we can show using Lemma~\ref{lemma:minpuns-gen}
that, if $i$ keeps defecting its neighbours, then
the expected number of punishments grows $1-2\epsilon$
every $(G,i,\rho)$-indistinguishable round, while decreasing at most $\epsilon$ in other rounds,
for an arbitrarily small $\epsilon$.
This means that, every period of $c$ rounds,
the expected number of punishments grows
$$\epsilon' = 1-2\epsilon - (c-1) \epsilon = 1- (c+1) \epsilon.$$
Since $c$ is constant on $\rho$,
we can find $\rho$ large enough and $\epsilon$ small enough
such that $\epsilon < 1/(c+1)$,
in which case $\epsilon' > 0$. This shows that the expected number 
of punishments is unbounded, as we intended to prove.
\end{proof}

We can now conclude with the proof of Theorem~\ref{theorem:dist-gen}.

\begin{restatable}{theorem}{theoremdistgen}
\label{theorem:dist-gen}
If the adversary is not restricted by frequent distinguishability,
then there is no safe protocol that enforces accountability in general pairwise exchanges.
\end{restatable}
\begin{proof}
By Lemma~\ref{lemma:acum-gen},
if the adversary is not restricted by frequent distinguishability,
then for an arbitrarily large $\rho$, there is $G$ and $i$ such that 
the expected number of punishments of $i$ 
in the $\rho$ rounds following any round is unbounded.
In particular, for an arbitrarily large $\rho$, there is $I_i \in \I_i(G)$
such that the expected number of punishments in the $\rho$
rounds after $I_i$ is larger than the number of $i$-edges,
and $i$ interacts with some neighbour $j$ at $I_i$,
so $i$ cannot be punished for defecting $j$ at $I_i$ in the next $\rho$ rounds.
This is a contradiction to Lemma~\ref{lemma:minpuns-gen},
concluding the proof.
\end{proof}
}

\subsection{A \oape{\G^*} for General Pairwise Exchanges with Connectivity}
\label{sec:general}

We introduce a restriction of connectivity on $\G^*$ that ensures that the adversary is restricted by eventual distinguishability
and is sufficient for devising safe protocols that enforce accountability in general pairwise exchanges.
To understand this condition, it is useful to first recall the scenario of Figure~\ref{fig:indist}.
In this scenario, the problem arises because neither agent 2 can distinguish $G^1$ from $G^3$ at round $3$
nor agent 3 can distinguish $G^2$ from $G^3$ at round $3$, and thus they cannot coordinate their actions to punish 1 only once.
This problem can be avoided if agents can learn the degree of their neighbours prior to deciding whether to punish them.
However, the knowledge of the degree is not sufficient to satisfy eventual distinguishability,
since for instance we may have a scenario similar to the one depicted in Figure~\ref{fig:indist} where agents 2 and 3
do not interact with 1 in the same round and cannot communicate with each other to coordinate the punishments.
This can be avoided if it is always the case that either 2 and 3 interact in the same round with 1 or
the first agent to interact with 1 can causally influence the other without interference from 1.
More generally, it suffices that, for every $G \in \G^*$, agent $i$, and round $m$,
there is $\rho$ such that the round-$m$ neighbours of $i$ causally influence every round-$m+\rho$ neighbour of $i$
between $m$ and $m+\rho$ without interference from $i$.
This condition is exactly met with $\rho = n$ when $\G^*$ is restricted by a 
condition similar to 1-connectivity from~\cite{Kuhn:10}:
we say that the adversary is \emph{restricted by connectivity} iff agents know the degree of their neighbours
and, for every $G \in \G^*$, agent $i$, and round $m$, the graph obtained from $G^m$ by removing the edges to $i$ is connected.
This condition is also met by overlays for gossip dissemination such as~\cite{Li:06,Li:08,Guerraoui:10}.

We now define a safe-bounded protocol $\vsigma^{\gen}$ that enforces accountability in
general pairwise exchanges, assuming that the adversary is restricted by connectivity.
Fix $G \in \G^*$ and let $\dgr_i^m$ denote the degree of $i$ in $G^m$.
At every round $m$, neighbouring agents always exchange
monitoring information, and they follow punishment individual actions 
with a probability proportional to past deviations.
Monitoring information includes reports and numbers of pending punishments.
Specifically, for each round $m' < m$ and pair of agents $(j,l)$, 
a report relative to $m'$ and $(j,l)$ specifies whether $j$ interacted with 
$l$ in round $m'$, and whether $l$ defected $j$.
For each $c \in \{1 \ldots n\}$ and agent $j$, 
agents keep the number of pending punishments to be applied to $j$ in periodic rounds $(c+ kn)_{k \geq 0}$.
Before interacting with $j$ in round $m$, $i$ determines whether $j$ should be punished.
For this, $i$ updates the number $x$ of pending punishments for the period that includes $m$.
If $m \leq n$, then $x = 0$. Otherwise, let $x'$ be the previous number of pending 
punishments resulting from an identical update prior to round $m-n$.
Given the round-$m-n$ reports, $i$ determines $\dgr_j^{m-n}$.
Then, $i$ sets $x$ to $\max(0,x'-\dgr_j^{m-n})$,
and adds $\dgr_j^{m-n}$ iff $j$ defected some neighbour in $m-n$. 
$i$ punishes $j$ with probability $\min(1,x/\dgr_j^m)$. (This is where the
knowledge of degree comes into play.)
After the interaction, $i$ emits a report indicating that the interaction
occurred in round $m$ and signalling whether $j$ defected $i$.
$i$ also updates its monitoring information basing on the information 
sent by its neighbours. For each report relative to round $m' < m$ not older than $m- n + 1$
and pair $(k,l)$, if $i$ does not have a report relative to $m'$ and $(k,l)$
and receives a new report from $j \neq l$, then $i$ stores this report.
In addition, for each period $c$ not including $m$, $i$ updates
the number of pending punishments relative to every $l\neq i$ and $c$
to the maximum between its value and the value sent by every neighbour $j \neq l$, capping it to lie in $\{0 \ldots n-1\}$.

\fullv{
We present in Alg.~\ref{alg:prop-pun} the pseudo-code for the strategy $\sigma_i^{\gen}$ of agent $i$.
We use two variables: $\pend{}{}$ and $\acc{}{}{}$. For each round $m$, $i$ keeps a number $\pend{j}{m}$ of pending 
punishments to be applied to $j$ in rounds $m,m+n,m+2n,\ldots$, initially equal to $0$ and never larger than $n-1$.
Also, $i$ keeps a report $\acc{j}{l}{m} \in \{\good,\bad,\bot\}$ per pair of nodes $(j,l)$ and round $r$ signalling whether in round $m$: (i) $j$ did not interact with $l$ ($\bot$); 
(ii) $l$ interacted and defected $j$ ($\bad$); or (iii) $l$ interacted and did not defect $j$ ($\good$). $\acc{j}{l}{m}$ is initialized to $\bot$.
Notice that these variables can be implemented by a finite state machine:
at each round $m$, agents only need to store and forward information relative to each tuple $(i,j,c)$ for $c \in \{0 \ldots n-1\}$,
corresponding to $\acc{i}{j}{m-c}$, and information relative to each pair $(i,c)$  for $c \in \{0 \ldots n-1\}$, corresponding to $\pend{i}{m-c}$.
For the sake of exposition, we opt to not represent the state in this compact form.

\begin{figure}[!t]
\begin{algorithm}[H]
\caption{$\sigma_i^{\gen}$}
\label{alg:prop-pun} 
{
\scriptsize
\begin{algorithmic}[1]
   \ForAll{$j, m$}
   	\State{$\pend{j}{m} \gets 0$}
   	\ForAll{$l$}
		\State{$\acc{j}{l}{m} \gets \bot$}
   	\EndFor
   \EndFor
   
\Statex
\UponEvent{round $m$}
	\ForAll{neighbour $j$}
		\State{$pr \gets \min(1,\pend{j}{m}/\dgr_j^m)$}
		\State{With probability $pr$, cooperate}\label{line:prop-pun:pun}\Comment{Punishment}
		\State{Otherwise, punish}\label{line:prop-pun:coop}\Comment{Cooperation}
	\EndFor
\EndEvent

\Statex

\After{round $m$}
	\ForAll{neighbour $j$}
		\If{$j$ defects $i$ in $m$}\label{line:prop-pun:update}
			\State{$\acc{i}{j}{m} \gets \bad$}
		\Else
			\State{$\acc{i}{j}{m} \gets \good$}
			\ForAll{$m' \in \{m - n +1 \ldots m-1\}$}
				\State{$\pend{j}{m'}$ $\gets$ maximum between $\pend{j}{m'}$ and values sent by each $l \neq j$}
				\State{Cap $\pend{j}{m'}$ to lie in $\{0 \ldots n - 1\}$}
				\ForAll{$l \ne i$}
					\If{$\acc{l}{j}{m'} = \bot$ and some $k \neq j$ sent $\acc{l}{j}{m'}|_{k} \neq \bot$}
						\State{$\acc{l}{j}{m'}$ $\gets$ $v\neq \bot$ deterministically selected among received values}
					\EndIf
				\EndFor
			\EndFor
		\EndIf\label{line:prop-pun:update:end}
	\EndFor
	\ForAll{agent $j \ne i$}
		\If{$m \ge n$}\label{line:prop-pun:up-pun}
			\State{$\dgr \gets \#\{l \ne j \mid \acc{l}{j}{m-n+1} \neq \bot\}$}\label{line:prop-pun:deg}\Comment{Degree $\dgr_j^{m-n+1}$}
			\State{$\pend{j}{m+1} \gets \max(0,\pend{j}{m-n+1} - \dgr)$}
    			\If{exists $l \ne j$ such that $\acc{l}{j}{m-n+1}=\bad$}
				\State{$\pend{j}{m+1} \gets \pend{j}{m+1} + \dgr$}
			\EndIf
		\EndIf
	\EndFor
\EndAfter
\end{algorithmic}
}
\end{algorithm}
\end{figure}

}

This definition has the following properties when agents follow $\vsigma^{\gen}$
at every round-$m$ history $h$:
(1) $i$ cannot influence monitoring information that determines punishments applied to $i$,
(2) we match each defection of $i$ in round $m$ to at least one punishment in future rounds,
(3) a defection in round $m$ triggers additional punishments
to be applied in rounds $m+n,m+2n\ldots$,
and (4) the delay of additional punishments is bounded by $O(n^2)$.
(1) ensures that $i$ does not gain from lying about monitoring information.
(2) ensures that, even if $i$ saves the cost $1$ of sending messages, $i$ loses at least $\beta > 1$.
(3) and (4) guarantee that this loss is discounted to the present by a lower bounded factor $\delta^{n^2}$.
This implies that, if agents are sufficiently patient (i.e., $\delta$ is sufficiently close to $1$), then $i$ prefers not to defect.

Theorem~\ref{theorem:prop-pun} shows that $\vsigma^{\gen}$
enforces accountability with connectivity and sufficiently patient agents.
\shortv{See Appendix~\ref{app:prop-pun} for the proof and a discussion about complexity.}

\begin{restatable}{theorem}{firstproppunoape}
\label{theorem:prop-pun}%
If the adversary is restricted by connectivity and agents are sufficiently patient, then 
$\vsigma^{\gen}$ enforces accountability in general pairwise exchanges.
\end{restatable}
\fullv{
\begin{proof}
We show that $\vsigma^{\gen}$ is a \oape{\G^*} in general pairwise exchanges.
Fix $G \in \G^*$, agent $i$, round-$m$, round-$m$ information set $I_i \in \I_i(G)$, $h \in I_i$,
and actions $a_i^*,a_i'$ such that $\sigma_i^{\gen}(a_i^*|I_i) > 0$.
Let $\Delta = \eu{i}{\vsigma^1}{h,G} - \eu{i}{\vsigma^2}{h,G}$,
where $\vsigma^1 = \vsigma^{\gen}|_{I_i,a_i^*}$, $\vsigma^2 = \vsigma^{\gen}|_{I_i,a_i'}$,
and $\vsigma^{\gen}|_{I_i,a_i}$ represents the protocol that differs from $\vsigma^{\gen}$
exactly in that $i$ deterministically follows $a_i$ at $I_i$.
\shortv{In the full paper,
we show that the result known
as One-Shot-Deviation principle\,\cite{Hendon:96} applies to the solution concept \oape{\G^*}.
This principle states that}\fullv{By the One-Shot-Deviation principle,} it suffices to show that $\Delta \geq 0$ to prove the result.
For every $m' \ge m$, let $\Delta^{m'}$
be the difference between the expected utility of $i$ in round $m'$,
between agents following $\vsigma^1$ and $\vsigma^2$, i.e.,
$$\Delta^{m'}  =  \sum_{r} \pra{r}{\vsigma^1}{G,h} u_i(\va^{m'}) - 
 \sum_{r} \pra{r}{\vsigma^1}{G,h} u_i(\va^{m'}),$$
where $\va^{m'}$ is the round-$m'$ action profile in $r$.
We prove six facts first, which hold for every run of $\vsigma^1$ and $\vsigma^2$ in $G$, conditioned on $h$.
We say that a fact holds after $i$ follows $\va_i'$ and $\va_i^*$
if the fact holds for all runs $r^1$ and $r^2$ of $\vsigma^1$ and $\vsigma^2$,
respectively, such that $\pra{r^1}{\vsigma^1}{G,h} > 0$ and $\pra{r^2}{\vsigma^2}{G,h} > 0$.
Fix agents $j$ and $l$ different from $i$, and rounds $m',m''$.
Let $\acc{i}{j}{m}|_{l,m'}$ and $\pend{i}{m}|_{l,m'}$
be the values of $\acc{i}{j}{m}$ and $\pend{i}{m}|_{l,m'}$, respectively,
deterministically held by $l$ at the end of round $m'$ in every run of the considered strategies:

\begin{enumerate}[(a)]
  \item Fact 1: If $i$ follows $a_i'$, then $\acc{i}{j}{m}|_{l,m'}$ is accurate: 
  (1) if $i$ and $j$ do not interact or it does not hold that $(j,m) \leadsto_i^G (l,m'+1)$, then the report says that they did not interact;
  (2) if $i$ and $j$ interact and $(j,m) \leadsto_i^G (l,m'+1)$, the the report says that $i$ defected $j$ iff $i$ defected $j$ in $a_i'$.
  \begin{proof}
  At the end of round $m$, the report of $l$ says that $i$ and $j$ interacted iff $l = j$ and $i$ and $j$ are neighbours.
  This is exactly the case when $(j,m) \leadsto_i^G (l,m+1)$. Moreover, (2) follows directly by construction of $\vsigma^{\gen}$.
  In every subsequent round $m'$, agent $l$ updates $\acc{i}{j}{m}|_{l,m'}$ to $\acc{i}{j}{m}|_{l,m'+1}$ only if $\acc{i}{j}{m}|_{l,m'} = \bot$
  and some round-$m'$ neighbour $o$ sent $\acc{i}{j}{m}|_{o,m'} \ne \bot$. By the hypothesis, $\acc{i}{j}{m}|_{o,m'}$ is accurate,
  so it follows that $\acc{i}{j}{m}|_{l,m'+1}$ is also accurate.
  \end{proof}
  
  \item Fact 2: If $m' = m+n - 1$ and $i$ follows $a_i'$, then $\pend{i}{m'+1}|_{l,m'} = y+ \max(x - \dgr_i^m,0)$,
  where $x$ is the maximum of $\pend{i}{m}|_{o,m}$ for all agents $o \ne i$,
  and $y = \dgr_i^m$ if $i$ defects a neighbour in $a_i'$ or $y = 0$ otherwise.
  \begin{proof}
  It is easy to show using induction that for every round $m'' \in \{m \ldots m'\}$,
  $\pend{i}{m}|_{l,m'} = x^{m''}$, where $x^{m''}$ is the maximum of $\pend{i}{m}|_{o,m}$
  for all agents $o \ne i$ such that $(o,m) \leadsto_i^G (l,m''+1)$. 
  This is clearly true for $m'' = m$, since only $l$ is causally influenced by $l$ between $m$ and $m+1$.
  For every other $m''$, $l$ updates $\pend{i}{m}|_{l,m''}$ to the maximum of the values
  sent by round-$m''-1$ neighbours (Lines~\algref{alg:prop-pun}{line:prop-pun:update}-\ref{line:prop-pun:update:end}),
  which by the induction hypothesis is $x^{m''}$.
  Since the adversary is restricted by connectivity,
  $l$ is causally influenced by every agent of $i$ without interference from $i$ between $m$ and $m+n$,
  so $\pend{i}{m}|_{l,m'} = x$. After updating $\pend{i}{m}|_{l,m'}$, $l$ computes the value $\dgr$,
  which by Fact~1 is $\dgr_i^m$, sets $\pend{i}{m'+1}|_{l,m'} = \min(0,x- \dgr_i^m)$,
  and adds $\dgr_i^m$ if some report indicates that $i$ defected some neighbour,
  which by Fact~1 is true iff $i$ defected some neighbour in $a_i'$. This proves Fact~2.
  \end{proof}
  
  \item Fact 3: If $m'' \neq m$ and $m' \ge m$, then the values of $\pend{i}{m''}|_{l,m'}$ 
  and $\acc{i}{j}{m''}|_{l,m'}$ when $i$ follows $a_i'$ are the same as when $i$ follows $a_i^*$.
  
  \begin{proof}
  If $m' = m$ and $m'' < m$, then those values depend only on the 
  values received by $l$ from agents different from $i$, which are determined by the history $h$.
  Continuing inductively, for every $m' > m$., $l$ updates $\pend{i}{m''}|_{l,m'}$ and $\acc{i}{j}{m}|_{l,m'}$ 
  according to the values sent by round-$m'$ neighbours $o\ne i$ of $l$, which are the same by the hypothesis,
  so the resulting values of $\pend{i}{m''}|_{l,m'}$ and $\acc{i}{j}{m}|_{l,m'}$ are the same, whether $i$ follows $a_i'$ or $a_i^*$.   
  If $m'' > m$, then $l$ does not update $\pend{i}{m''}|_{l,m'}$ and $\acc{i}{j}{m}|_{l,m'}$ when $m' < m''$.
  When $m' = m''-1$, $l$ updates $\pend{i}{m''}|_{l,m'}$ basing on the 
  values $\pend{i}{m''-n}|_{l,m'}$ and $\acc{i}{j}{m''-n}|_{l,m'}$,
  which are the same by the proof of the case $m'' < m$.
  When $m' = m''$, $l$ updates $\acc{i}{j}{m''-n}|_{l,m'}$ iff $l$ interacts with $j$ in round $m''$,
  and updates it to $\good$, whether $i$ follows $a_i'$ and $a_i^*$, since $i$ follows $\sigma_i^{\gen}$
  after $m$ and does not defect $l$.
  Finally, when $m' > m''$, the same arguments as above prove the result.  
  \end{proof}
  
  \item Fact 4: If $m' > m$, then the value $\pend{i}{m''}|_{l,m'}$ is at least as high
  when $i$ follows $a_i'$ as when $i$ follows $a_i^*$.
  \begin{proof}
  First, consider that $m'' < m$. By Fact~3, the values are the same whether $i$ follows $a_i'$ or $a_i^*$.
  Consequently, the same holds for all rounds $m'' + cn$ with $c>0$.
  Regarding $m''=m$, by Fact~2, the value $\pend{i}{m''}|_{l,m'}$ is higher by 
  $\dgr_i^m$ when $i$ follows $a_i'$ than $a_i^*$ if $i$ defects some neighbour in $a_i'$, or is the same otherwise.
  Again, the result follows directly from Fact~3 for all rounds $m+cn$ with $c>0$.
  \end{proof}
  
  \item Fact 5: If $m' > m$, then $\Delta^{m'} \ge 0$.
  \begin{proof}
  In round $m'$, both $i$ and its neighbours either cooperate or punish each other, so the
  benefits and communication costs are the same, whether $i$ follows $a_i'$ or $a_i^*$.
  By Fact~4, $i$ is never punished with higher probability when $i$ follows $a_i^*$
  than when $i$ follows $a_i'$, hence the expected costs of punishments must be at least as high
  when $i$ follows $a_i'$ as when it follows $a_i^*$. This implies that $\Delta^{m'} \ge 0$. 
  \end{proof}
  
  \item Fact 6: If $i$ defects some neighbour in $a_i'$, then $\sum_{m < m^* < m+n^2} \Delta^{m^*} \ge \beta \dgr_i^m$.
  \begin{proof}
  Let $m^c = m + cn$ for $c\ge 0$, and let $\pend{i}{m^c}|_{o,m^c-1,a_i^*}$ and $\pend{i}{m^c}|_{o,m^c-1,a_i'}$
  be the value of $\pend{i}{m^c}|_{o,m^c-1}$ when $i$ follows $a_i^*$ and $a_i'$, respectively.
  By Fact~2, if $i$ defects some neighbour in $a_i'$, we have
  $\pend{i}{m^1}|_{o,m^1-1,a_i'} - \pend{i}{m^c}|_{o,m^1-1,a_i^*} > \deg_i^m$.
  For every $c \ge 0$ with $\pend{i}{m^{c}}|_{o,m^{c}-1,a_i^*} \ge \dgr_i^{m^c}$, we have 
  $$\pend{i}{m^{c+1}}|_{o,m^{c+1}-1,a_i^*} -  \pend{i}{m^{c}}|_{o,m^c-1,a_i^*} = \dgr_i^{m^c}.$$
  Let $m^* = m^{c}$ be the first round where $\pend{i}{m^{c}}|_{o,m^{c}-1,a_i^*} \le \dgr_i^{m^c}$.
  When $i$ follows $a_i^*$, in round~$m^*$, $i$ is punished by an expected number of $\pend{i}{m^{*}}|_{o,m^*-1,a_i^*}$
  neighbours, while not being punished afterwards in rounds $m^* + cn$ for $c > 0$.
  When $i$ follows $a_i'$, $i$ is punished in rounds $m^* + cn$ for $c \ge 0$ by an 
  additional expected number of $\dgr_i^m$ neighbours.
  Moreover, since $i$ always has at least one neighbour and $i$ never defects after $m$, by Fact~2
  and the fact that $\pend{i}{m+n}|_{o,m+n-1,a_i'} \le n-1$, we have $\pend{i}{m^n}|_{o,m^n-1,a_i'} = 0$.
  This implies that the additional punishments when $i$ follows $a_i'$ must occur in rounds
  between $m+1$ and $m+n^2$. By the same arguments as in the proof of Fact~5,
  it follows that the expected utility of $i$ decreases at least by $\beta \dgr_i^m$ in these rounds when $i$ follows $a_i'$
  instead of $a_i^*$, as we intended to prove.
  \end{proof} 
\end{enumerate}

We can now show that $\Delta \ge 0$. If $i$ does not defect a neighbour in $a_i'$,
then $\Delta^{m} \ge 0$ and by Fact~5 $\Delta \ge 0$.
If $i$ defects neighbours in $a_i'$, then $\Delta^{m} \ge -\dgr_i^m$,
and by Fact~6, we have $\Delta \ge -\dgr_i^m + \delta^{n^2} \beta \dgr_i^m$.
So, if $\beta > 1$ and $\delta$ is sufficiently close to $1$, then $\Delta \ge 0$.
This concludes the proof.
\end{proof}

}

\fullv{
\subsection{Avoiding Knowledge of Degree}

We now discuss one type of interactions in general pairwise exchanges
and corresponding implementation of $\vsigma^{\gen}$, 
assuming only that agents know the identities of their neighbours.
We need at least two communication phases,
so that neighbours $i$ and $j$ may first reveal their degree in phase 1, and
then exchange values in phase 2 and punish each other accordingly.
Unfortunately, two phases is not enough because $i$ may lie about the degree.
In particular, in phase 1, $i$ may declare a higher degree than the real one to decrease the probability of
being punished by $j$. We can address this by including in the report
information about the declared degrees, and then punish agents that lie about their degrees.
This is still not sufficient though, because of the following scenario.
Suppose that $i$ only has one neighbour $j$ and only one pending punishment.
If $i$ does not lie, then $j$ punishes $i$ with probability $1$ due to
having a single pending punishment. If $i$ declares a degree of $n-1$ instead,
then the probability of being punished by $j$ is only $1/(n-1)$.
This yields an increase in the expected benefits from $0$ to $\beta(n-2)/(n-1)$.
Suppose that $i$ also defects $j$ by omitting messages in phase 2.
The expected future loss must be at most $\beta$, to avoid the problem
identified in the proof of need for eventual distinguishability. Since we only assume that $\beta > 1$, 
the loss may be lower than the gain.
The problem is that $i$ sends the degree before incurring the cost of sending messages in phase 2.
We can avoid this by using a technique of delaying gratification employed in~\cite{Li:06}.
We need three communication phases.
In phase 1, agents exchange monitoring information and the values 
ciphered with random private keys. In phase 2, they reveal
the degrees. Finally, in phase 3, they decide whether to cooperate by sending the private keys,
or to punish by sending arbitrary keys. Let $\alpha^{\kappa}$ be the cost of sending a key,
such that the cost of sending phase 1 and 2 messages plus $\alpha^{\kappa}$ is $1$.
Agent $i$ is punished for defecting $j$ in phases 1 or 2 by not receiving the value in phase 3.
In this case, $i$ avoids at most $1 + \alpha$, but loses $\beta > 1+\alpha$ (notice that,
by consistency of beliefs, $i$ always believes that it will receive the value in phase 3 prior to this phase,
even if $j$ does not send its value correctly ciphered,
because $i$ cannot distinguish ciphered values from arbitrary ones).
The maximum utility gain obtained by $i$ is when $i$ has only one neighbour $j$,
$i$ lies to $j$ by saying that its degree is $n-1$, and then defect $j$ in phase 3.
The maximum gain is $\beta(n-2)/(n-1) + \alpha^\kappa$,
whereas the future loss is $\delta \beta$.
If $\alpha^\kappa < \beta/(n-1)$, then the gain is less than $\beta$,
and the future loss outweighs the gain for a $\delta$ sufficiently close to $1$.
Therefore, we have a protocol that enforces accountability.

\subsection{Complexity}

The bit complexity of $\vsigma^{\gen}$ is $O(n^3)$:
each message carries $n^3$ reports, with one report per pair of agents and round in $1 \cdots \rho$,
carries $n^2$ accusations, and carries $n^2$ numbers of pending punishments.
The maximum delay of punishments is $O(n^2)$.
We can improve the complexity by assuming further restrictions on $\G^*$
and on the computational ability of agents.
Specifically, the factors $n^3$ and $n^2$ are a function of 
(1) the maximum delay $n$ of disseminating information relative to an agent, 
and (2) the maximum number $n$ of agents relative to which an agent has to forward information.
(1) can be improved by considering more restrictive assumptions about $\G^*$.
As discussed in the definition of adversary restricted by connectivity,
we only need that the neighbours of any given agent $i$ in round $m$ can causally influence
all the neighbours of $i$ in round $m+\rho$ for some constant $\rho$.
If the evolving graphs satisfy locality properties that ensure that the neighbours of agent $i$'s neighbours
are likely to be $i$'s neighbours in the near future, then $\rho$ can be significantly smaller than $n$.
Then, the bit complexity becomes $O(\rho n^2)$ and the maximum delay of punishments becomes $O(\rho n)$.
Small-world networks such as social networks, which have a high clustering coefficient,
are well known examples that satisfy such locality properties\,\cite{Holland:71,Watts:98}.

Interestingly, locality properties also allow us the improve on (2).
Agents only have to forward information  relative to an agent $i$ and round $m$
in the $\rho$ rounds following $m$. This is necessary to ensure that critical information
reaches the agents capable of carrying a punishment. If during every period of $\rho$
rounds each agent were only causally influenced by a limited number $c$ of agents,
then agents only needed to forward information relative to $c$ agents,
and the bit complexity and delay of punishments would be $O(\rho c^2)$ and $O(\rho c)$, respectively.
Unfortunately, such reduction on the amount of information  each agent forwards introduces a congestion problem.
Specifically, we cannot let agents sending messages of varying size, since that would incentivize
them to always forward the minimum required information to save bandwidth.
Therefore, they must forward messages of fixed size, which implies that agents
may have to discard information if they cannot fit all received information in a single message.
This gives the opportunity for an agent $i$ to deliberately generate false reports (or accusations)
that flood the network, causing other agents to discard accusations against $i$.
We can address this issue by appending to each report and accusation a signature of the issuer.
A report of an interaction between $i$ and $j$ has issuer $i$, so a valid report of this interaction
must contain a signature of $i$. The same applies to accusations and numbers of pending punishments.
The only disadvantage of this approach is that it requires agents to be computationally bounded.

}

\section{Discussion}
\label{sec:disc}

Our results provide technical insight on the type of dynamic networks
that can support pairwise exchanges in equilibria strategies. The need
for timely punishments means that accountability cannot be enforced in
some dynamic networks such as certain overlays for file-sharing.
If connectivity is ensured, accountability may be enforced, even for general
exchanges. Although for this scenario we have assumed knowledge of
degree prior to interactions, this can be relaxed if agents can
exchange multiple messages per round\shortv{ (appendix~\ref{app:prop-pun}
elaborates on this idea)}. Finally, if exchanges are valuable, timely
punishments are enough to enforce accountability, which opens the door
to protocols that enforce accountability in a wide variety of dynamic networks.
There are multiple open questions to be addressed in future work.
It would be interesting to prove a stronger condition
that would close the gap between an adversary restricted by distinguishability
and one restricted by connectivity. Another open issue is collusion. 
We believe that both the necessary and sufficient conditions presented in this paper could
be strengthened by generalizing the notion of causal influence
without interference from individual agents to the absence of interference
from members of a coalition. Given these conditions, the \oape{\G^*}
strategies are resilient to collusion. 


\bibliography{disc16}
\bibliographystyle{plain}

\appendix

\shortv{
\section{Additional Definitions}
We define some notation used in the proofs in appendix, and we define
the notion of consistent beliefs.

\subsection{Additional Notation}
Given a round-$m$ history $h$ and $G \in \G$, a protocol $\vsigma$ defines a probability
distribution $\pra{r}{\vsigma}{G,h}$ over every run $r$ compatible
with $h$ and $G$ (i.e., $r(m) = h$ and $G$ is the evolving graph in
$h$). Let $\pra{h}{\vsigma}{G}$ be the probability of history $h$
being realized conditioned on $G$ and agents following $\vsigma$ from
the beginning.  We say that $r$ is a run of $\vsigma$ after $h$ in $G$ if
$\pra{r(m')}{\vsigma}{G,h} > 0$ for every round $m' > m$.
As an abuse of notation, we consider that $\pra{r}{\vsigma}{G,h} > 0$.
We say that $r$ is a run of $\vsigma$ in $G$ if $h$ is the initial history.

\subsection{Consistent Beliefs}
\label{app:consistent}

The definition of consistent beliefs is identical in spirit to that proposed in the definition of sequential equilibrium\,\cite{Kreps:82}.
We say that a belief system $\mu^*$ is 
consistent with $\vsigma^*$ and $\G^*$ iff there is a sequence $\vsigma^1,\vsigma^2,\ldots$ 
of completely mixed protocols (attribute positive probability to every possible action at every information set) converging to $\vsigma^*$
such that for every $G \in \G^*$, agent $i$, information set $I_i \in \I_i$, and $h \in \I_i$,
$$\mu(h \mid G,I_i) = \lim_{c \to \infty} \pra{h}{\vsigma^c}{G}/\sum_{h' \in I_i} \pra{h'}{\vsigma^c}{G}.$$

In the full paper\,\cite{Vilacatr:16a}, we prove two properties of consistent beliefs used in the analysis.
One is that, when an agent $i$ observes an information set consistent with only $i$ deviating,
$i$ believes that only $i$ deviated. Another is that when $i$ expects to be punished
$x$ times after observing round-$m$ information set $I_i$, then there is a round-$m+1$
information set $I_i'$ at which $i$ expects to be punished at least $x-k$ number of times,
where $k$ is the number of round-$m$ neighbours of $i$.
}

\shortv{
\section{Proofs of Necessary Conditions}
\label{app:nec}
We prove the need for timely punishments and eventual distinguishability.

\subsection{Need for Timely Punishments}
\label{app:citi}

Recall that an adversary is said to be restricted by timely punishments
iff there exists a constant $\rho > 0$ such that, for all $G \in \G^*$, agent $i$, 
and $i$-edge $(j,m)$, there is a PO $(l,m')$ of $i$ for $(j,m)$ in $G$ with
$m' < m + \rho$. We prove Theorem~\ref{theorem:citi}.

\theoremciti*

\begin{proof}
The proof is by contradiction. Suppose that $\vsigma^*$ enforces accountability.
Suppose also that the adversary is not restricted by timely punishments.
For every $\rho > 0$, we can fix $G \in \G^*$, agent $i$, and $i$-edge $(j,m)$
such that, for every PO $(l,m')$ of $i$ for $(j,m)$ in $G$, we have $m' - m \ge \rho$.
This implies that, for every round $m' \ge m$ with $m' < m+\rho$,
and every $i$-edge $(l,m')$ with $l \ne j$, $(j,m) \leadsto^G (l,m')$ is false, since the
existence of a PO $(l,m')$ of $i$ for $(j,m)$ in $G$ with $m' < m+\rho$ is equivalent
to the existence of one such $(l,m')$ with $(j,m) \leadsto^G (l,m')$.
Let $\rho$ be such that $y \delta^{\rho}/(1-\delta) < 1$,
where $y$ is the (bounded) maximum difference between utilities of a single interaction.
Since $\sigma_i^*$ enforces accountability, $i$ is expected to cooperate at 
every round-$m$ information set $I_i \in \I_i(G)$ consistent with $\vsigma^*$ and $G$.
Here, $i$ may follow a single evasive strategy $\sigma_i'$ 
relative to $\vsigma^*$, $G$, and $(j,m)$ \shortv{(See Appendix~\ref{app:evasive})}. 
Let $\vsigma' = (\sigma_i',\vsigma^*_{-i})$.
By Lemma~\ref{lemma:hide}, for every round $m' < m+\rho$ with $m' \ge m$,
the round-$m'$ neighbours of $i$ observe each round-$m'+1$ information
set with the same probability, whether $i$ follows $\sigma_i'$ or $\sigma_i^*$.
Since these information sets specify the individual actions followed by and taken towards $i$,
the expected utility of $i$ in round $m'$ is the same,
whether agents follow $\vsigma'$ or $\vsigma^*$,
except $i$ avoids at least the cost $1$ of defecting $j$ in round $m$.
Moreover, the maximum utility difference in every round $m' \ge m+\rho$ is $y n$.
Therefore, we have
$$u_i(\vsigma^* \mid G,I_i) - u_i(\vsigma' \mid G,I_i) < -1 + \delta^{\rho} yn/(1-\delta) < 0.$$

This contradicts the assumption that $\vsigma^*$ enforces accountability, concluding the proof.
\end{proof}
}

\shortv{
\subsection{Need for Eventual Distinguishability}
\label{app:dist}

\theoremdist*
\begin{proof}
The proof is by contradiction. Suppose that the adversary is not restricted by eventual distinguishability.
By Lemma~\ref{lemma:acum},
there is $G$ and $i$ such that
for every $c > 0$, there is a round $m$ and round-$m$ $I_i \in \I_i(G)$ such that
$\mathbb{E}^{\vsigma^*,\rho}[P_i \mid G,I_i] \ge c$.
This is true for $m$ large enough such that
$c$ is larger than the number of $i$-edges between $m$ and $m+\rho$.
This is a contradiction to $\vsigma^*$ converging to a cooperation state in $\rho$ rounds.
\end{proof}
}

\shortv{
\section{\oape{\G^*} for Valuable Pairwise Exchanges with Timely Punishments}
\label{app:penances}
We prove that the protocol $\vsigma^{\val}$ from Section~\ref{sec:valuable} 
enforces accountability in valuable pairwise exchanges,
and we discuss different implementations of valuable pairwise exchanges.

\subsection{Analysis}

\theorempenance*

}

\shortv{
\section{\oape{\G^*} for General Pairwise Exchanges with Connectivity}
\label{app:prop-pun}
We provide the pseudo-code for the protocol $\vsigma^{\gen}$ from Section~\ref{sec:general} and a proof that it enforces
accountability in general pairwise exchanges. We conclude with a discussion
on how to avoid knowledge of the degree.

\subsection{Pseudo-code}

\subsection{Analysis}

We now prove Theorem~\ref{theorem:prop-pun}.

\firstproppunoape*

}

\shortv{

\subsection{Avoiding Prior Knowledge of Degree}

}

\shortv{

\subsection{Complexity}

}

\shortv{
\section{Evasive Strategies}
\label{app:evasive}

}

\end{document}